\documentclass{revtex4}

\usepackage{latexsym}
\usepackage{graphics}
\usepackage{amsmath}
\usepackage{amssymb}
\usepackage{gensymb}
\usepackage[usenames,dvipsnames]{color}
\usepackage{epstopdf}
\usepackage{bm}
\usepackage{amsthm}

\newtheorem{condition}{Condition}
\newtheorem{assumption}{Assumption}
\newtheorem{remark}{Remark}
\newtheorem{theorem}{Theorem}
\newtheorem{lemma}{Lemma}

\begin{document}
\title{The Smoluchowski-Kramers limit of stochastic differential equations with arbitrary state-dependent friction}

\author{Scott Hottovy}
\affiliation{ University of Wisconsin, Department of Mathematics, Madison, Wisconsin 53706 USA, \\ University of Arizona, Department of Mathematics, Tucson, Arizona 85721 USA}
\author{Austin McDaniel}
\affiliation{University of Arizona,  Department of Mathematics, Tucson, Arizona 85721 USA}
\author{Giovanni Volpe}
\affiliation{Soft Matter Lab, Department of Physics, Bilkent University, Ankara 06800, Turkey, UNAM-Institute of Material Science and Nanotechnology, Bilkent University, 06800 Ankara, Turkey}
\author{Jan Wehr}
\affiliation{University of Arizona,  Department of Mathematics, Tucson, Arizona 85721 USA}


\begin{abstract}
We study a class of systems of stochastic differential equations describing diffusive phenomena.  The Smoluchowski-Kramers approximation is used to describe their dynamics in the small mass limit. Our systems have arbitrary state-dependent friction and noise coefficients.  We identify the limiting equation {and,} in particular, the additional drift term that appears in the limit is expressed in terms of the solution to a Lyapunov matrix equation.  The proof uses a theory of convergence of stochastic integrals developed by Kurtz and Protter. The result is sufficiently general to include systems driven by both white and Ornstein-Uhlenbeck colored noises. We discuss applications of the main theorem to several physical phenomena, including the experimental study of Brownian motion in a diffusion gradient. 
\end{abstract}

\maketitle

\section{Introduction}
\label{sec:intro}

For an open subset $\mathcal{U}\subset\mathbb{R}^d$, consider the $2d$-dimensional stochastic differential equation (SDE):
\begin{equation}\label{eq:SDEgeneral}
\left\{\begin{array}{rclrcl}
d\bm{x}_t^m & = & \bm{v}_t^m\,dt  &\bm{x}_0^m &=& \bm{x}, \\
d\bm{v}_t^m & = & \left[ \frac{\bm{F}(\bm{x}_t^m)}{m} - \frac{\bm{\bm{\gamma}}(\bm{x}_t^m)}{m}\bm{v}^m_t \right] \,dt + \frac{\bm{\sigma}(\bm{x}_t^m)}{m}\,d\bm{W}_t \quad &\bm{v}_0^m &=& \bm{v},
\end{array}\right.
\end{equation} 
with $\bm{F}:\mathcal{U}\mapsto\mathbb{R}^d$, $\bm{\gamma}:\mathcal{U}\rightarrow \mathbb{R}^{d\times d}$ a $d\times d$ invertible matrix-valued function, $\bm{\sigma}:\mathcal{U}\rightarrow \mathbb{R}^{d\times k}$ and $\bm{W}$ a $k$-dimensional Wiener process.  The above SDE provides a framework to model many physical systems, from colloidal particles in a fluid {\cite{nelson}} to a camera tracking an object \cite{papanicolaou2010}. For example, the motion of a Brownian particle can be modeled using an SDE where $x$ and $v$ are one-dimensional and ${\gamma}(x) = {k_BT \over D(x)}$ and ${\sigma}(x) = {k_BT\sqrt{2} \over \sqrt{D(x)}}$ (see description below in Section~\ref{sec:BP}).  In fact, the original motivation for the present work was to provide a mathematical explanation of the experimental observation of a \emph{noise-induced drift} in \cite{volpe2010}. While in this model the coefficients $\gamma(x)$ and $\sigma(x)$ are constrained by the fluctuation-dissipation relation such that ${\gamma}(x) \propto {\sigma}(x)^2$ \cite{kubo}, our main result, Theorem~\ref{theorem}, does not assume it and has a much more general reach including applications in other fields.

Theorem~\ref{theorem} says that, under the assumptions stated in Section~\ref{sec:SKa}, the $\bm{x}$-component of the solution of equation~(\ref{eq:SDEgeneral}) converges in $L^ 2$, with respect to the topology on $C_{\mathcal{U}}([0,T])$ (i.e. the space of continuous functions from $[0,T]$ to $\mathcal{U}$ with the uniform metric), to the solution of the SDE
\begin{equation}\label{eq:SKlimit}
d\bm{x}_t=\left[ \bm{\gamma}^{-1}(\bm{x}_t)\bm{F}(\bm{x}_t)+\bm{S}(\bm{x}_t) \right] dt + \bm{\gamma}^{-1}(\bm{x}_t)\bm{\sigma}(\bm{x}_t)d\bm{W}_t,
\end{equation}
with the original initial condition $\bm{x}_0 = \bm{x}$, where $\bm{S}(\bm{x}_t)$ is the \emph{noise-induced drift} whose $i^\text{th}$ component equals
\begin{equation}
\label{eq:spurious}
S_i(\bm{x}) =\frac{\partial}{\partial x_{l}}[(\gamma^{-1})_{ij}(\bm{x})]J_{j{l}}(\bm{x}),
\end{equation}
where $\bm{J}$ is the matrix solving the Lyapunov equation
\begin{equation}\label{eq:Lyapunov}
\bm{J}\bm{\gamma}^* + \bm{\gamma}\bm{J} = \bm{\sigma}\bm{\sigma}^*.
\end{equation}
Throughout the paper we use Einstein summational convention and ``$^*$" denotes the transposition of a matrix. The limiting SDE~(\ref{eq:SKlimit}) is given in the It\^o form, while we provide in Section~\ref{sec4} the corresponding Stratonovich form. Note that for $m > 0$ the process $\bm{x}^m_t$ has bounded variation and thus all definitions of stochastic integral lead to the same form of SDE~(\ref{eq:SDEgeneral}).

The zero-mass limits of equations similar to equation~(\ref{eq:SDEgeneral}) have been studied by many authors beginning with Smoluchowski \cite{smoluchowski1916} and Kramers \cite{kramers1940}. In the case where $F=0$ and $\gamma$ and $\sigma$ are constant, the solution to equation~(\ref{eq:SDEgeneral}) converges to the solution of equation~(\ref{eq:SKlimit}) almost surely \cite{nelson}.  The case including an external force was treated by entirely different methods in \cite{schuss}. The problem of identifying the limit for position-dependent noise and friction was studied in \cite{hanggi1982} for the case when the fluctuation-dissipation relation is satisfied and in \cite{sancho1982} for the general one-dimensional case (the multidimensional case is also discussed there but without complete proof). The homogenization techniques described in \cite{papanicolaou1975,schuss,pavliotis} were used to compute the limiting backward Kolmogorov equation as mass is taken to zero in \cite{hottovy2012}. In \cite{Pardoux2003} convergence in distribution is proven rigorously for equations of the same type as equation~(\ref{eq:SDEgeneral}), under somewhat stronger assumptions than those made here. The rigorous proof of convergence in probability for $\bm{\gamma}$ constant and $\bm{\sigma}$ position-dependent is given in \cite{freidlin2004}. The present paper contains the first rigorous derivation of the zero-mass limit of equation~(\ref{eq:SDEgeneral}) for a multidimensional system with general friction and noise coefficients.

Systems with colored noise can also be treated within the above (suitably adapted) framework. For example, the one-dimensional equation driven by an Ornstein-Uhlenbeck (OU) noise with a short correlation time $\tau$
\begin{equation}\label{eq:Newton}
m \ddot{x}^m_t =F(x_t^m)-\gamma(x^m_t) \dot{x}^m_t + \sigma(x_t^m){\eta}_t^\tau
\end{equation}
can be rewritten in the form of equation~(\ref{eq:SDEgeneral}), by defining $\bm{v}_t^m = (v_t^m,\eta_t^\tau)^*$, $\bm{x}_t^m = (x_t^m,\zeta_t^\tau)^*$ and $\tau = \tau_0m$ \cite{pavliotis}, as
\begin{equation}\label{eq:SDEOUCN}
\left \{ \begin{array}{rcl}
dx_t^m &=& v_t^m\,dt \\
dv_t^m &=& \left[ \frac{F(x_t^m)}{m}-\frac{\gamma(x_t^m)}{m} v_t^m + \frac{\sigma(x_t^m)}{m}\eta_t^\tau \right] dt \\
d\zeta_t^\tau &=& \eta_t^\tau dt \\
d\eta_t^\tau &=& -\frac{a \eta_t^\tau}{\tau}\,dt + \frac{\sqrt{2\lambda}}{\tau}\,dW_t .
\end{array} \right.
\end{equation}
The details are worked out in Section~\ref{sec:OU}. SDE with colored noise were first studied in \cite{wong1965}, where it was shown that, as the correlation time of the noise goes to zero, the stochastic integral converges to the Stratonovich integral with respect to the Wiener process. This result was generalized in \cite{kurtz91} and similar systems were studied in \cite{kupferman2004} by homogenization methods. Our method is sufficiently general to permit us to recover the results obtained in these works as well as those in \cite{wong1965,freidlin2011}.

In Section~\ref{sec:SKa} we state and prove the main result, Theorem~\ref{theorem}, for arbitrary dimension $d$. In Section~\ref{sec:1D} we explicitly calculate the limit for $d=1$. In Section~\ref{sec5} we present a series of applications of our result. In Section~\ref{sec:BP} we study the equations describing the experiment on Brownian motion in a diffusion gradient that originally motivated this work \cite{volpe2010}. In Section~\ref{sec:OU} we study the case of SDE driven by OU colored noise and find the explicit limit for constant (Section~\ref{sec:constfric}) and position-dependent (Section~\ref{sec:thermo}) friction.  In Section~\ref{sec:3DBP} we study a three-dimensional Brownian particle on which a non-conservative external force is acting, and in Section~\ref{sec:magnetic} we consider the more specific case of a magnetic force. In Section 5 we reformulate the main result using Stratonovich formalism.

\begin{acknowledgements}
We thank Thomas Kurtz for the crucial references \cite{kurtz91} and \cite{blount1991}, and Krzysztof Gaw\c edzki for pointing out some earlier results. AM was partially supported by the NSF grant DMS 1312711.  JW was partially supported by NSF grants DMS 1009508 and DMS 1312711. SH was partially supported by the NSF under grant DMS 1009508 and grant No. DGE0841234. GV was partially supported by the Marie Curie Career Integration Grant (MC-CIG) No. PCIG11 GA-2012-321726.
\end{acknowledgements}

\section{Smoluchowski-Kramers approximation}\label{sec:SKa}

For the main theorem, we assume $\bm{x}_t^m,\bm{x}_t\in\mathcal{U}\subset\mathbb{R}^d$, an open set, and $\bm{v}_t^m\in\mathbb{R}^d$ for all $0\leq t\leq T$. For an arbitrary vector $\bm{a}\in\mathbb{R}^d$, $|\bm{a}|$ is the Euclidean norm and, for a $d\times d$ matrix $\bm{A}\in\mathbb{R}^{d\times d}$, $|\bm{A}|$ is the induced operator norm. We now state the assumptions and main theorem. 
\begin{assumption}\label{assume:bddcoeffs}
The coefficients $\bm{F},\bm{\gamma},\bm{\sigma}$ are continuously differentiable functions. Furthermore, the smallest eigenvalue $\lambda_1(\bm{x})$ of the symmetric part ${1\over2} (\bm{\gamma}+\bm{\gamma}^*)$ of the matrix $\bm{\gamma}$ is positive uniformly in ${\bm x}$, i.e. 
\begin{equation}\label{eq:AssumeEig}
\lambda_1(\bm{x}) \ge c_\lambda > 0.
\end{equation}
It follows that $(\bm{\gamma}(\bm{x}) \bm{y}, \bm{y}) \geq c_{\lambda} (\bm{y}, \bm{y})$ and $|\bm{\gamma}(\bm{x})|\geq c_{\lambda}$ for all $\bm{x}\in\mathcal{U}, \bm{y} \in \mathbb{R}^d$ and that the real parts of the eigenvalues of $\bm{\gamma}(\bm{x})$ are bounded below by $c_\lambda$.
\end{assumption}
\begin{remark}
The lower bounds on $\bm{\gamma}$ and its eigenvalues are crucial for the estimates of the proof. A system with vanishing friction, i.e. $\bm{\gamma}(\bm{x})={\bm 0}$, behaves differently \cite{freidlin2013}. 
\end{remark}
\begin{assumption}\label{assume:existence}
With probability one there exist global unique solutions, {defined on $[0,T]$},  to equation~(\ref{eq:SDEgeneral}) for each $m$ and to equation~(\ref{eq:SKlimit}). In particular, there are no explosions.  
\end{assumption}
\begin{assumption}\label{assume:stop}
With probability one there exists a {compact set} $\mathcal{K}\subsetneq\mathcal{U}$ such that, for all $m>0$, $\bm{x}_t^m\in\mathcal{K}$ for all $t \in [0, T]$.
\end{assumption}
The existence of such a set $\mathcal{K}$, together with the continuity of the coefficients $\bm{F}$, $\bm{\gamma}$ and $\bm{\sigma}$, implies that there exists a constant $C_T$, depending only on $T$ (in particular, independent of $m$), such that 
\begin{equation}
|\bm{F}(\bm{x}_t^m)|, \, |\bm{\sigma}(\bm{x}_t^m)|, \, |\bm{\gamma}(\bm{x}_t^m)|\leq C_T,
\end{equation} 
for all  $t \in [0, T]$. 

\begin{theorem}\label{theorem}
Suppose SDE~(\ref{eq:SDEgeneral}) satisfies Assumptions~1-3. Let $(\bm{x}_t^m,\bm{v}_t^m)\in\mathcal{U}\times\mathbb{R}^{d}$ be its solution with initial condition $(\bm{x},\bm{v})$ independent of $m$ and let $\bm{x}_t$ be the solution to the It\^o SDE~(\ref{eq:SKlimit}) with the same initial position $\bm{x}_0 = \bm{x}$.  Then
\begin{equation}
\lim_{m\rightarrow 0} E\left [\left (\sup_{0\leq t\leq T}|\bm{x}_t^m-\bm{x}_t|\right )^2\right ]=0.
\end{equation} 
\end{theorem}

Before proving the theorem, we state a lemma about convergence of stochastic integrals, which restates (in a slightly less general form) a theorem of Kurtz and Protter \cite{kurtz91}.

\subsection{Convergence of Stochastic Integrals}\label{sec:convergenceSI}

Let $\{\mathcal{F}_t\}_{t \geq 0}$ be a filtration on a probability space $(\Omega,\mathcal{F},P)$.  We assume that it  satisfies  the usual conditions \cite{revuz}. In our case, $\mathcal{F}_t$ will be (the usual augmentation of) $\sigma(\{\bm{W}_s:s\leq t\})$, the $\sigma$-algebra generated by a $k$-dimensional Wiener process $\bm{W}_t$ up to time $t$.  

Suppose $\bm{H}$ is an $\{\mathcal{F}_t\}$-adapted semi-martingale with paths in $C_{\mathbb{R}^n}[0,T]$, whose  Doob-Meyer decomposition is $\bm{H}_t = \bm{M}_t + \bm{A}_t$ so that $\bm{M}_t$ is an $\mathcal{F}_t$-local martingale and $\bm{A}_t$ is a process of locally bounded variation \cite{revuz}.  For a continuous $\{\mathcal{F}_t\}$-adapted process $\bm{X}$ with paths in $C_{\mathbb{R}^{d\times n}}[0,T]$ and for $t \leq T$ consider the It\^o integral 
\begin{equation}\label{eq:defint}
\int_0^t \bm{X}_s\,d\bm{H}_s = \lim \sum_{i}\bm{X}_{t_i}(\bm{H}_{t_{i+1}} - \bm{H}_{t_i}),
\end{equation}
where $\{t_i\}$ is a partition of $[0,t]$ and the limit is taken as the maximum of $t_{i+1}-t_i$ goes to zero. For a continuous processes $\bm{X}_s$ such that 
\begin{equation}
P\left ( \int_0^T |\bm{X}_s|^2 \,d\langle \bm{M} \rangle_s + \int_0^T |\bm{X}_s|\,dV_s(\bm{A}) <\infty\right ) = 1, 
\end{equation}
where $\langle \bm{M} \rangle_s$ is the quadratic variation of $\bm{M}_s$ and $V_s(\bm{A})$ is the total variation of $\bm{A}_s$, the limit in equation~(\ref{eq:defint}) exists in the sense that 
$$
\sup_{0 \leq t \leq T}\left (\left | \int_0^t \bm{X}_s\,d\bm{H}_s - \sum_{i}\bm{X}_{t_i}(\bm{H}_{t_{i+1}} - \bm{H}_{t_i})\right |\right )\rightarrow 0,
$$
in probability.  This (and related) convergence modes will be used throughout the paper \cite{protter,revuz}. 

Consider $(\bm{U}^m,\bm{H}^m)$ with paths in $C_{\mathbb{R}^d\times\mathbb{R}^n}[0,T]$ adapted to $\{\mathcal{F}_t\}$ where $\bm{H}^m_t$ is a semi-martingale with respect to $\mathcal{F}_t$. Let $\bm{H}^m_t = \bm{M}_t^m + \bm{A}_t^m$ be its Doob-Meyer decomposition.  Let $\bm{f}:\mathcal{U}\rightarrow \mathbb{R}^{d\times n}$ be a continuous {matrix-valued} function and let $\bm{X}^m$, with paths in $C_{\mathcal{U}}[0,T]$, satisfy the SDE 
\begin{equation}\label{eq:KPthm1}
\bm{X}_{t}^{m} = \bm{X}_0 + \bm{U}_{t}^m + \int_0^{t} \bm{f}(\bm{X}_s^{m})\,d\bm{H}_s^m,
\end{equation}
where $\bm{X}_0^m = \bm{X}_0 \in\mathbb{R}^d$ is the same initial condition for all $m$. Define $\bm{X}$, with paths in $C_{\mathcal{U}}[0,T]$, to be the solution of
\begin{equation}\label{eq:KPlimit}
\bm{X}_t = \bm{X}_0 + \int_0^tf(\bm{X}_s)\,d\bm{H}_s.
\end{equation}
Note that (\ref{eq:KPthm1}) implies $\bm{U}_0^m = \bm{0}$ for all $m$. 

\begin{lemma}\cite[Theorem 5.10]{kurtz91}\label{theorem:KP} 
Assume $(\bm{U}^m,\bm{H}^m)\rightarrow (\bm{0},\bm{H})$ in probability with respect to $C_{\mathbb{R}^d\times\mathbb{R}^n}[0,T]$, i.e. for all $\epsilon>0$, 
\begin{equation}\label{eq:defofprob}
P\left [ \sup_{0\leq s \leq T} \left( |\bm{U}_s^m|+|\bm{H}_s^m-\bm{H}_s| \right) >\epsilon\right ]\rightarrow 0,
\end{equation}
as $m\rightarrow 0$, and the following condition is satisfied: 
\begin{condition}\label{condition}[Tightness condition]
The total variations, $\{V_t(\bm{A}^m)\}$, are stochastically bounded for each $t>0$, i.e. $P[V_t(\bm{A}^m)>L]\rightarrow 0$ as $L\rightarrow\infty$, uniformly in $m$.
\end{condition}
Suppose that there exists a unique global solution to equation~(\ref{eq:KPlimit}). Then, as $m\rightarrow0$, $\bm{X}^m$ converges to $\bm{X}$, the solution of equation~(\ref{eq:KPlimit}), in probability with respect to $C_{\mathcal{U}}([0,T])$.
\end{lemma}

To cast the limiting equation in the form of Lemma 1, it would be enough to rewrite equation~(\ref{eq:SDEgeneral}) and check that Condition~\ref{condition} is satisfied.  In our case, the limiting equation is 
\begin{equation}
d\bm{x}_t=\left[ \bm{\gamma}^{-1}(\bm{x}_t)\bm{F}(\bm{x}_t)+\bm{S}(\bm{x}_t) \right] dt + \bm{\gamma}^{-1}(\bm{x}_t)\bm{\sigma}(\bm{x}_t)d\bm{W}_t, \quad \bm{x}_0 = \bm{x}.
\end{equation}
To state the limiting equation, it would be enough to define
$$\bm{f}(\bm{x}) = (\bm{\gamma}^{-1}(\bm{x})\bm{F}(\bm{x}),  \bm{\gamma}^{-1}(\bm{x})\bm{\sigma}(\bm{x}), \bm{S}(\bm{x})).$$
However, to describe the equations with $m >0$ using the same function $\bm{f}$, we need it to have more components. In the limit $m \to 0$ these additional components will be integrated against zero processes and thus will not contribute to the stochastic integral. That is, we will apply Lemma~\ref{theorem:KP}, with $\bm{f}$ of the form 
\begin{equation}
\bm{f}(\bm{x}) = (\bm{\gamma}^{-1}(\bm{x})\bm{F}(\bm{x}),  \bm{\gamma}^{-1}(\bm{x})\bm{\sigma}(\bm{x}), \bm{S}(\bm{x}), ... ),
\end{equation}
where $\bm{f}$ contains more components and the limit process $\bm{H}_t$ has zeros in the corresponding rows, i.e. 
\begin{equation}
\bm{H}_t = \begin{pmatrix} t \\ \bm{W}_t \\ t \\ 0 \\  \vdots \\ 0 \end{pmatrix}. 
\end{equation}
Note also that $\bm{H}_t$ has two components equal $t$ to separate the noise-induced drift $\bm{S}$ from the term $\bm{\gamma}^{-1}\bm{F}$.

\begin{proof}[of Theorem~\ref{theorem}]

We first state and prove a lemma about the convergence of the processes $m\bm{v}^m$ to zero.

\begin{lemma}\label{lemma:convergenceVas}
For each $m>0$, let $\bm{x}_t^m$ be any $\mathcal{F}_t$-adapted process with continuous paths in $\mathcal{K}$  and define $\bm{v}_t^m$ as the solution to the SDE given by the second equation in~(\ref{eq:SDEgeneral}) with functions $\bm{F}$, $\bm{\gamma}$, and $\bm{\sigma}$ satisfying Assumptions~\ref{assume:bddcoeffs}-\ref{assume:stop}. Then, for any $p \geq 1$, $m\bm{v}^m\rightarrow 0$ as $m\rightarrow 0$ in $L^p$ with respect to $C_{\mathbb{R}^d}[0,T]$, and hence also in probability with respect to $C_{\mathbb{R}^d}[0,T]$, i.e. 
\begin{equation}
\label{lemma2assertion}
\lim_{m\rightarrow 0} E\left[ \left (\sup_{0\leq t \leq T} |m\bm{v}_t^m| \right )^p \right] =0.
\end{equation}
and, for all $\epsilon>0$,
\begin{equation}
\lim_{m\rightarrow 0} P\left (\sup_{0\leq t \leq T} |m\bm{v}_t^m|>\epsilon \right )=0.
\end{equation}
\end{lemma}
\begin{proof}
Consider the SDE for $m\bm{v}_t^m$ given by equation~(\ref{eq:SDEgeneral}), 
\begin{equation}
	d(m\bm{v}_t^m) = \bm{F}(\bm{x}_t^m)\;dt -\frac{\bm{\gamma}(\bm{x}_t^m)}{m}(m\bm{v}_t^m)\;dt + \bm{\sigma}(\bm{x}_t^m)\;
d\bm{W}_t. 
\end{equation}
This equation is similar to an Ornstein-Uhlenbeck equation, which we would obtain with $\bm{F} = 0$ and $\bm{\gamma}$ and $\bm{\sigma}$ constant.  Thus we use this similarity to bound $m\bm{v} ^m$.  We use a technique similar to the proof of Lemma 3.19 in 
\cite{blount1991}.  We first define the function 
\begin{equation}
	f_m(u) = \frac{2m}{c_{\lambda}} \int_0^{\sqrt{c_{\lambda} u/(2 m \Gamma )}}e^{s^2/2}\int_0^s e^{-r^2/2}\;dr\;ds,
\end{equation}
where $\Gamma = C_T^2d$ ($C_T$ is the bound from Assumption~3 and $d$ is the dimension of $\bm{v}_t^m$, i.e. the dimension of the space). 
Note that $f_m(0) = 0$, $f'_m(u),f''_m(u)>0$ for all $u\in[0,\infty)$.  Also, $f_m(u)\rightarrow\infty$ and $f_m(m^2u)\rightarrow0$ as $m\rightarrow0$ for all $u>0$.  Furthermore, 
\begin{equation}
\label{eq:Aidentity}
Af_m(u) = 1
\end{equation}
 for all $u\in[0,\infty)$,
where $A$ is the differential operator defined by
\begin{equation}
\label{eq:Aequation}
Af_m(u) \equiv f'_m(u)(-\frac{c_\lambda}{m}u+ 2 \Gamma) + 4 \Gamma u f''_m(u)
\end{equation}
  We will prove that
\begin{equation}
	P\left (\sup_{0\leq t \leq T}|m\bm{v}_t^m|^2\geq \epsilon \right ) \leq 
	\frac{f_m\left (|m\bm{v}|^2\right ) +T}{f_m(\epsilon)}\rightarrow 0, 
\end{equation}
as $m\rightarrow 0$.  Using the It\^o product formula for $|m\bm{v}^m_t|^2 = m(\bm{v}_t^m)^*m\bm{v}^m_t$, we obtain 
\begin{align}
	d(m(\bm{v}_t^m)^*m\bm{v}_t^m) =& m(\bm{v}_t^m)^*d(m\bm{v}_t^m) + d(m\bm{v}_t^m)^*m\bm{v}_t^m + 
	d(m\bm{v}_t^m)^*d(m\bm{v}_t^m) \\
\label{eq:mv2diff}
	=&-\frac{2}{m}(\bm{\gamma}(\bm{x}_t^m)m\bm{v}_t^m,m\bm{v}_t^m)\;dt \\
	+& Tr(\bm{\sigma} (\bm{x}_t^m) \bm{\sigma} ^* (\bm{x}_t^m))\;dt + m(\bm{v}_t^m)^* \bm{F}(\bm{x}_t^m)dt +  \bm{F}(\bm{x}_t^m)^*m\bm{v}_t^mdt \nonumber \\
	+& m(\bm{v}_t^m)^* (\bm{\sigma}(\bm{x}_t^m)\;d\bm{W}_t) + (\bm{\sigma}(\bm{x}_t^m)\;d\bm{W}_t)^* m\bm{v}_t^m. \nonumber
\end{align}

  By the It\^o formula for all $t\in[0,T]$, 
\begin{align}
	f_m\left(|m\bm{v}_t^m|^2\right ) = &f_m\left (|m\bm{v}|^2\right ) \\
	+& \int_0^t \left [f'_m\left (|m\bm{v}_s^m|^2 \right)\Big (-\frac{2}{m}(\bm{\gamma}(\bm{x}_s^m)m\bm{v}_s^m,m\bm{v}_s^m) \right .\nonumber\\
	+& \left .m(\bm{v}_s^m)^* \bm{F}(\bm{x}_s^m) +  \bm{F}(\bm{x}_s^m)^*m\bm{v}_s^m + Tr(\bm{\sigma} (\bm{x}_s^m) \bm{\sigma} ^* (\bm{x}_s^m)) \Big) \right .\nonumber\\
	 +& 2 f''_m\left (|m\bm{v}_s^m|^2 \right) |m\bm{\sigma} ^* (\bm{x}_s^m) \bm{v}^m_s|^2 \Big ] \;ds
	 + M_t \nonumber
\end{align}
where $M_t \in C_{\mathbb{R}} [0,T]$ is a martingale with $E[M_t]=0$. Next we use the bound, 
\begin{align}
	(m\bm{v}_t^m,\bm{F}(\bm{x}_t^m)) \leq& \frac{1}{2}|m\bm{v}_t^m|^2 + \frac{1}{2}|\bm{F}(\bm{x}_t^m)|^2
\end{align}
 and from Assumption~1
\begin{equation}
	(\bm{\gamma}(\bm{x}_t^m)m\bm{v}_t^m,m\bm{v}_t^m) \geq c_{\lambda}|m\bm{v}_t^m|^2 . 
\end{equation}
Using $f' _m (u), f'' _m (u)>0$ for all $u\in[0,\infty)$, we obtain
\begin{align}
f_m\left(|m\bm{v}_t^m|^2\right) \leq & f_m\left (|m\bm{v}|^2\right )+\int_0^t \Big[f'_m\left(|m\bm{v}_s^m|^2\right )\Big(-\frac{2c_\lambda}{m}|m\bm{v}_s^m|^2  \\
+& |m\bm{v}_s^m|^2 + |\bm{F}(\bm{x}_s^m)|^2 + Tr(\bm{\sigma} (\bm{x}_s^m) \bm{\sigma} ^* (\bm{x}_s^m)) \Big)\nonumber \\
+& 2 f''_m\left(|m\bm{v}_s^m|^2\right)|m\bm{v}_s^m|^2 | \bm{\sigma} (\bm{x}_s^m) |^2 \Big ] \;ds
	 + M_t . \nonumber 
\end{align}
For small $m>0$, the first term under the integral will dominate the second. 
More precisely, for  $\bm{x}_s^m$ in the compact set $\mathcal{K}$ and for $m$ sufficiently small so that $\frac{c_\lambda}{m} > 1$, we have
\begin{align}
	f_m\left (|m\bm{v}_t^m|^2\right ) \leq &  f_m\left (|m\bm{v}|^2\right ) \\
	+&\int_0^t [f'_m\left (|m\bm{v}_s^m|^2\right )\Big(-\frac{c_\lambda}{m}|m\bm{v}_s^m|^2 +C_T ^2 + C_{T}^2d\Big)\nonumber \\
+& 2 f''_m\left (|m\bm{v}_s^m|^2\right )|m\bm{v}_s^m|^2 C_T^2\Big ] \;ds
	 + M_t . \nonumber
\end{align}
Using the definition of $\Gamma$ and equations ~(\ref{eq:Aequation}) and ~(\ref{eq:Aidentity}) we get
\begin{align}
	f_m(|m\bm{v}_t^m|^2) \leq & f_m\left (|m\bm{v}|^2\right )+\int_0^t [f'_m(|m\bm{v}_s^m|^2)(-\frac{c_\lambda}{m}|m\bm{v}_s^m|^2 +2 \Gamma) \\
+&  4 \Gamma |m\bm{v}_s^m|^2 f''_m(|m\bm{v}_s^m|^2) \Big ] \;ds
	 + M_t\nonumber\\
	 = & f_m\left (|m\bm{v}|^2\right )+\int_0^t Af_m(|m\bm{v}_s^m| ^2)\;ds + M_t \\
	 = &f_m\left (|m\bm{v}|^2\right )+ t + M_t.
\end{align}
Define $\tau_\epsilon^m=\inf\{t: |m\bm{v}_t^m|^2=\epsilon\}$. Then for all 
$\epsilon>0$, 
\begin{equation}
	P\left ( \sup_{0\leq t\leq T} |m\bm{v}_t^m|^2\geq\epsilon \right ) = 
	P\left ( |m\bm{v}_{T\wedge \tau_\epsilon^m}^m|^2\geq\epsilon\right ). 
\end{equation} 
Next, because $f_m$ is an increasing function (since $f' _m (u)>0$ for all $u\geq0$), 
\begin{equation}
	P\left ( |m\bm{v}_{T\wedge \tau_\epsilon^m}^m|^2\geq\epsilon\right ) =
	P\left ( f_m( |m\bm{v}_{T\wedge \tau_\epsilon^m}^m|^2)\geq f_m(\epsilon)\right ) 
\end{equation}
Finally we use Chebyshev's inequality and the Optional Stopping Theorem to obtain, 
\begin{align}
P\left ( \sup_{0\leq t\leq T} |m\bm{v}_t^m|^2\geq\epsilon \right ) \leq& \frac{E[f_m(|m\bm{v}_{T\wedge \tau_\epsilon^m}^m|^2)]}{f_m(\epsilon)} \leq \frac{E[f_m\left (|m\bm{v}|^2\right )+T\wedge \tau_\epsilon^m]}{f_m(\epsilon)}\\
\leq& \frac{f_m\left (|m\bm{v}|^2\right ) +T}{f_m(\epsilon)}. 
\end{align}
 Recalling that $f_m \left (m^2|\bm{v}|^2\right ) \rightarrow 0$ and $f_m(\epsilon) \rightarrow \infty$ as $m \rightarrow 0$, this inequality proves that as $m\rightarrow 0$, $ \sup_{0 \leq t \leq T} | m\bm{v}_t^m|^2 \rightarrow 0$ in probability, i.e.,
for all $\epsilon>0$,
\begin{equation}
\label{eq:convergenceinprob}
\lim_{m\rightarrow 0} P \left (\sup_{0\leq t \leq T} |m\bm{v}_t^m| ^2 > \epsilon \right ) =0.
\end{equation}
We prove that $m\bm{v} ^m$ converges to zero in $L^p$ with respect to $C_{\mathbb{R}^d}[0,T]$.  Let $q > 1$, then
\begin{align*}
E\left[ \left (\sup_{0\leq t \leq T} |m\bm{v}_t^m|^2 \right )^q \right] &= \int _0 ^{\infty} q x ^{q - 1} P\left ( \sup_{0\leq t\leq T} |m\bm{v}_t^m|^2 \geq x \right ) dx \\
&\leq \int _0 ^{\infty} q x ^{q - 1} \frac{f_m\left (|m\bm{v}|^2\right ) +T}{f_m(x)} dx \\
&\leq q(1 + T) \int _0 ^{\infty}   \frac{x ^{q - 1}}{f_m(x)} dx 
\end{align*}
for $m$ sufficiently small since $f_m \left (|m \bm{v}|^2\right ) \rightarrow 0$ as $m \rightarrow 0$.  Since
\begin{align*}
f_m(x) &= \frac{2m}{c_{\lambda}} \int_0^{\sqrt{c_{\lambda} x/(2 m \Gamma )}}e^{s^2/2}\int_0^s e^{-r^2/2}\;dr\;ds \\
&\geq \frac{2m}{c_{\lambda}} \int_0^{\sqrt{c_{\lambda} x/(2 m \Gamma )}}e^{s^2/2} \left( \frac{s}{2} \right) e^{-s^2/8} \;ds \\
&= \frac{1}{4 \Gamma} \int_0^x e^{\frac{3 c_{\lambda} u}{16 m \Gamma}} du \; \geq  \; \frac{1}{4 \Gamma} \left( \frac{x}{2} \right) e^{\frac{3 c_{\lambda} x}{32 m \Gamma}}
\end{align*} 
it follows that
\begin{equation*}
E\left[ \left (\sup_{0\leq t \leq T} |m\bm{v}_t^m|^2 \right )^q \right] \leq C(q) < \infty
\end{equation*}
where $C(q)$ depends on $q$ but is independent of $m$.  Thus, there exists $m_0 > 0$ such that the family $\{ \sup_{0\leq t \leq T} |m\bm{v}_t^m| ^p \; : \; 0 < m \leq m_0 \}$ is uniformly integrable for $ 1 \leq p < 2q$ \cite[13.3]{williams}.  This fact together with (\ref{eq:convergenceinprob}) implies (\ref{lemma2assertion})\cite[13.7]{williams}.  Q.E.D.
\end{proof}

To determine the limit of SDE~(\ref{eq:SDEgeneral}) as $m\rightarrow 0$, we rewrite the equation for $\bm{v}_t^m$ as 
\begin{equation}
	\bm{\gamma}(\bm{x}^m_t)\bm{v}^m_t \,dt = \bm{F}(\bm{x}^m_t)\,dt + \bm{\sigma}(\bm{x}^m_t)d\bm{W}_t - md\bm{v}^m_t. 
\end{equation}
By Assumption~\ref{assume:bddcoeffs}, $\bm{\gamma}(\bm{x})$ is invertible, thus
\begin{equation}\label{eq:vdt}
d\bm{x}_t^m = \bm{v}_t^m\,dt = \bm{\gamma}^{-1}(\bm{x}^m_t) \bm{F}(\bm{x}^m_t)\,dt +
\bm{\gamma}^{-1}(\bm{x}^m_t) {\bm{\sigma}}(\bm{x}^m_t)d\bm{W}_t - m\bm{\gamma}^{-1}(\bm{x}^m_t)\,d\bm{v}^m_t,
\end{equation}
or, in integral form, 
\begin{equation}\label{eq:sub}
\bm{x}_t^m =\bm{x} + \int_0^t\bm{\gamma}^{-1}(\bm{x}^m_s) \bm{F}(\bm{x}^m_s)\,ds +\int_0^t\bm{\gamma}^{-1}(\bm{x}^m_s){\bm{\sigma}}(\bm{x}^m_s)
d\bm{W}_s - \int_0^t m\bm{\gamma}^{-1}(\bm{x}^m_s)\,d\bm{v}^m_s.
\end{equation}
In order to apply Lemma~\ref{theorem:KP} we need to integrate the last term by parts (see Remark~\ref{remark2}).

\subsection{Integration by parts to satisfy assumptions of Lemma~\ref{theorem:KP}} To determine the limit of the expression~(\ref{eq:sub}) as $m\rightarrow 0$, we consider its $i$th component. Integrating by parts the last term on the right-hand side of equation~(\ref{eq:sub}) we obtain, noting that $\bm{v}_0^m = \bm{v}$, 
\begin{align}\label{eq:subvvstar}
\int_0^t m\left[ (\gamma^{-1})_{ij}(\bm{x}^m_s)\right]\,d(v^m_s)_j =& (\gamma^{-1})_{ij}(\bm{x}^m_t)m({v}_t^m)_j - (\gamma^{-1})_{ij}(\bm{x})m{v}_j \\
-& \int_0^t \frac{\partial}{\partial x_{l}}[(\gamma^{-1})_{ij}(\bm{x}^m_s)]m({v}^m_s)_jd({x}^m_s)_{l}.\nonumber
\end{align}
Since $d(x^m_s)_{{l}} = (v^m_s)_{{l}}\,ds$, the last integral can be rewritten as
\begin{equation}
\int_0^t \frac{\partial}{\partial x_{l}}[(\gamma^{-1})_{ij}(\bm{x}^m_s)]m({v}^m_s)_j (v^m_s)_{{l}}\,ds.
\end{equation}
Note that $\bm{x}_t^m$ has bounded variation, hence the It\^o term in the integration by parts formula is zero. The product $m({v}^m_s)_j (v^m_s)_{{l}}$ in the above integral is the $(j,{l})$-entry of the (outer product) matrix $m\bm{v}^m_s(\bm{v}^m_s)^*$.  We will express this matrix as a solution of an equation.  To this end, we calculate, using the It\^o product formula,
\begin{equation}
d[m\bm{v}^m_s(m\bm{v}^m_s)^*] = \,d(m\bm{v}^m_s)(m\bm{v}^m_s)^* + m\bm{v}^m_s\,d(m\bm{v}^m_s)^* + d(m\bm{v}^m_s)\,d(m\bm{v}^m_s)^*.
\end{equation}
We now substitute for $md(\bm{v}^m_s)$ and for its adjoint the expression from equation~(\ref{eq:SDEgeneral}), obtaining
\begin{align}\label{eq:dmvmv}
d[m\bm{v}^m_s(m\bm{v}^m_s)^*] =&  \left[ m\bm{F}(\bm{x}_s^m)(\bm{v}^m_s)^* - m\bm{\gamma}(\bm{x}_s^m)\bm{v}^m_s(\bm{v}^m_s)^* \right] \,ds  \\
+& m\left (\bm{\sigma}(\bm{x}_s^m)\,d\bm{W}_s \right )(\bm{v}^m_s)^*\nonumber\\
+& \left[ m\bm{v}^m_s\bm{F}(\bm{x}_s^m)^* - m\bm{v}^m_s(\bm{v}^m_s)^*\bm{\gamma}^*(\bm{x}_s^m) \right] \,ds \nonumber\\
+&m\bm{v}^m_s\left (\bm{\sigma}(\bm{x}_s^m)\,d\bm{W}_s \right )^* + \bm{\sigma}(\bm{x}_s^m)\bm{\sigma}^*(\bm{x}_s^m)\,ds \nonumber.
\end{align}
Because of Lemma~\ref{lemma:convergenceVas}, we expect the terms proportional to $m\bm{v}_s^m$ to converge to zero in probability. Defining
\begin{equation}\label{eq:tildeU}
\tilde{\bm{U}}_t^m =  \int_0^t m{\bm{v}}^m_s{\bm{F}}^*(\bm{x}_s^m) ds + \int_0^tm{\bm{v}}^m_s({\bm{\sigma}}(\bm{x}_s^m)d\bm{W}_s)^*,
\end{equation}
we can rewrite equation~(\ref{eq:dmvmv}) as
\begin{align}\label{eq:lyapunov}
 -& m\bm{v}_t^m(\bm{v}_t^m )^*\bm{\gamma}^*(\bm{x}_t^m)dt  - \bm{\gamma}(\bm{x}_t^m)m \bm{v}_t^m(\bm{v}_t^m)^*dt \\ = & d[m\bm{v}_t^m(m\bm{v}_t^m)^*]-\bm{\sigma}(\bm{x}_t^m)\bm{\sigma}^*(\bm{x}_t^m)\,dt 
 - d\tilde{\bm{U}}_t^m - d(\tilde{\bm{U}}_t^m)^*.\nonumber
\end{align}
Equation~(\ref{eq:lyapunov}) can be written as 
\begin{align}\label{eq:lyapunov2}
&[m\bm{v}_t^m(\bm{v}_t^m )^*dt][-\bm{\gamma}^*(\bm{x}_t^m)] + [- \bm{\gamma}(\bm{x}_t^m)][m \bm{v}_t^m(\bm{v}_t^m)^*dt] \\
= & d[m\bm{v}_t^m(m\bm{v}_t^m)^*]-\bm{\sigma}(\bm{x}_t^m)\bm{\sigma}^*(\bm{x}_t^m)\,dt - d\tilde{\bm{U}}_t^m - d(\tilde{\bm{U}}_t^m)^*.\nonumber
\end{align}
Denoting $m\bm{v}_t^m(\bm{v}_t^m )^*dt$ by $\bm{V}$, $-\bm{\gamma}(\bm{x}_t^m)$ by $\bm{A}$ and the right-hand side of equation~(\ref{eq:lyapunov2}) by $\bm{C}$, we obtain
\begin{equation}\label{eq:lyapunovgeneral}
\bm{A}\bm{V} + \bm{V}\bm{A}^* = \bm{C},
\end{equation}
which is a Lyapunov equation \cite{ortega,bellman}. By \cite[Theorem 6.4.2]{ortega},  if the real parts of all eigenvalues of $\bm{A}$ are negative, then the Lyapunov equation has a unique solution, given by \cite[Chapter 11]{bellman}
\begin{equation}\label{eq:lyapunovanalytical}
\bm{V} = -\int_0^\infty e^{\bm{A}y}\bm{C}e^{\bm{A}^*y}\,dy. 
\end{equation}
By Assumption~1, this applies to $\bm{A} = -\bm{\gamma}(\bm{x}_t^m)$, giving
\begin{align}\label{eq:mvv}
m\bm{v}_t^m(\bm{v}_t^m)^*dt =& -\int_0^\infty e^{-\bm{\gamma}(\bm{x}_t^m)y}\left ( d[m\bm{v}_t^m(m\bm{v}_t^m)^*]-\bm{\sigma}(\bm{x}_t^m)\bm{\sigma}^*(\bm{x}_t^m)\,dt  \right. \\
-& \left . d\tilde{\bm{U}}_t^m - d(\tilde{\bm{U}}_t^m)^*\right )e^{-\bm{\gamma}^*(\bm{x}_t^m)y}\,dy \nonumber\\
=&\underbrace{-\int_0^\infty e^{-\bm{\gamma}(\bm{x}_t^m)y}d[m\bm{v}_t^m(m\bm{v}_t^m)^*]e^{-\bm{\gamma}^*(\bm{x}_t^m)y}\,dy}_{d\bm{C}^1_t}\nonumber \\
+& \underbrace{\int_0^\infty e^{-\bm{\gamma}(\bm{x}_t^m)y}\left (\bm{\sigma}(\bm{x}_t^m)\bm{\sigma}^*(\bm{x}_t^m)\,dt\right )\,e^{-\bm{\gamma}^*(\bm{x}_t^m)y}\,dy}_{d\bm{C}^2_t} \nonumber\\
+& \underbrace{\int_0^\infty e^{-\bm{\gamma}(\bm{x}_t^m)y}\left (d\tilde{\bm{U}}_t^m + d(\tilde{\bm{U}}_t^m)^*\right )e^{-\bm{\gamma}^*(\bm{x}_t^m)y}\,dy}_{d\bm{C}^3_t} . \nonumber
\end{align}
We will treat each term in a different way: after substituting the above expression into equation~(\ref{eq:subvvstar}), the term with $\bm{C}_t^1$ will be included in the $\bm{H}_t^m$ process (in the notation of Lemma~\ref{theorem:KP}), the $\bm{C}_t^2$ term will become a part of the noise-induced drift term $\bm{S}$ in the limiting equation~(\ref{eq:SKlimit}), and the $\bm{C}_t^3$ term will become a part of $\bm{U}_t^m$ which will be shown to converge to zero. 
For the first term, 
\begin{equation}
d(C^1_t)_{ij} = - \int_0^\infty (e^{-\bm{\gamma}(\bm{x}_t^m)y})_{i k_1}(e^{-\bm{\gamma}^*(\bm{x}_t^m)y})_{k_2 j}\,dy \,d[m({v}_t^m)_{k_1}(m{v}_t^m)_{k_2}^*],
\end{equation}
where the integral exists and is finite for all $t\in[0,T]$. 
For the second term, $d\bm{C}^2_t=\bm{J}(\bm{x}_t^m)dt$, where $\bm{J}(\bm{x}):\mathcal{U}\rightarrow\mathbb{R}^{d\times d}$ is the solution to the Lyapunov equation
\begin{equation}\label{eq:Glyapunov}
\bm{J}\bm{\gamma}^* + \bm{\gamma}\bm{J} = \bm{\sigma}\bm{\sigma}^*,
\end{equation} 	
as follows from differentiating the (Lebesgue) integrals in the identity
\begin{equation}
{\int_0^t[\bm{J}(\bm{x}_s^m)\bm{\gamma}^*(\bm{x}_s^m) + \bm{\gamma}(\bm{x}_s^m)\bm{J}(\bm{x}_s^m)]\,ds = \int_0^t \bm{\sigma}(\bm{x}_s^m)\bm{\sigma}^*(\bm{x}_s^m)\,ds. }
\end{equation}
For the third term, using the equation~(\ref{eq:tildeU}) for $\tilde{\bm{U}}^m$, the entries of $\bm{C}^3$ can be written as
\begin{align}
(\bm{C}^3_t)_{ij} =&  \int_0^t \int_0^\infty (e^{-\bm{\gamma}(\bm{x}_s^m)y})_{i k_1}\left ([m{\bm{v}}^m_s{\bm{F}}^*(\bm{x}_s^m)]_{k_1 k_2} ds  \right. \\ 
+&\left . [m{\bm{v}}^m_s({\bm{\sigma}}(\bm{x}_s^m)d\bm{W}_s)^*]_{k_1 k_2} 
+ [{\bm{F}}(\bm{x}_s^m) (m{\bm{v}}^m_s)^*]_{k_1 k_2}ds \right .\nonumber \\
+& \left . [{\bm{\sigma}}(\bm{x}_s^m)d\bm{W}_s(m{\bm{v}}^m_s)^*]_{k_1 k_2}\right )(e^{-\bm{\gamma}^*(\bm{x}_s^m)y})_{k_2 j}\,dy \nonumber \\
=& \sum_{k_1 k_2} \int_0^t \int_0^\infty(e^{-\bm{\gamma}(\bm{x}_s^m)y})_{i k_1}(e^{-\bm{\gamma}^*(\bm{x}_s^m)y})_{k_2 j}\,dy\left  ([m{\bm{v}}^m_s{\bm{F}}^*(\bm{x}_s^m)]_{k_1 k_2} ds  \right .\nonumber\\ 
+& \left .[m{\bm{v}}^m_s({\bm{\sigma}}(\bm{x}_s^m)d\bm{W}_s)^*]_{k_1 k_2} 
+ [{\bm{F}}(\bm{x}_s^m) (m{\bm{v}}^m_s)^*]_{k_1 k_2}ds \right.\nonumber \\ 
+&\left . [{\bm{\sigma}}(\bm{x}_s^m)d\bm{W}_s(m{\bm{v}}^m_s)^*]_{k_1 k_2} \right ).\nonumber
\end{align}

We substitute the expression for $m\bm{v}_t^m(\bm{v}_t^m)^*\,dt$ back into equation~(\ref{eq:subvvstar}). In the resulting formula for $\bm{x}_t^m$, the contribution from $\bm{C}^3$ will form a vector-valued process $\bm{U}^m$. Integrating equation~(\ref{eq:sub}) by parts and substituting equation~(\ref{eq:mvv}) for $(v_s^m)_j(v_s^m)_{l} ds$, 
\begin{align}
	\label{eq:xmcomplete}
	(&{x}_t^m)_i = {x}_i + ({U}_t^m)_i + \int_0^t (\bm{\gamma}^{-1}(\bm{x}_s^m)\bm{F}(\bm{x}_s^m))_i \,ds \\
+ &\left (\int_0^t (\bm{\gamma}^{-1}(\bm{x}_s^m)\bm{\sigma}(\bm{x}_s^m))d\bm{W}_s\right) _i \nonumber\\
	+& \int_0^t\frac{\partial}{\partial x_{l}}[(\gamma^{-1})_{ij}(\bm{x}^m_s)]J_{j{l}}(\bm{x}_s^m)\,ds \nonumber\\
	+& \int_0^t\frac{\partial}{\partial x_{l}}[(\gamma^{-1})_{ij}(\bm{x}^m_s)] \times \nonumber \\
&\left [-\int_0^\infty (e^{-\bm{\gamma}(\bm{x}_s^m)y})_{jk_1}(e^{-\bm{\gamma}^*(\bm{x}_s^m)y})_{k_2 l}\,dy \right ] d[(mv_s^m)_{k_1}(mv_s^m)_{k_2}],\nonumber
\end{align}
where $\bm{U}_t^m$ is
\begin{align}\label{eq:fullU}
(\bm{U}^m_t)_i =& (\gamma^{-1})_{ij}(\bm{x}^m_t)m({v}_t^m)_j - (\gamma^{-1})_{ij}(\bm{x})m{v}_j \\
+&\int_0^t\frac{\partial}{\partial x_{l}}[(\gamma^{-1})_{ij}(\bm{x}^m_s)]\times \nonumber\\
&\left [ \int_0^\infty(e^{-\bm{\gamma}(\bm{x}_s^m)y})_{jk_1}(e^{-\bm{\gamma}^*(\bm{x}_s^m)y})_{k_2 l}\,dy \times \right . \nonumber\\
& \left  ([m{\bm{v}}^m_s{\bm{F}}^*(\bm{x}_s^m)]_{k_1 k_2} ds 
 +[m{\bm{v}}^m_s({\bm{\sigma}}(\bm{x}_s^m)d\bm{W}_s)^*]_{k_1 k_2} \right . \nonumber\\
+ & \left .  [{\bm{F}}(\bm{x}_s^m) (m{\bm{v}}^m_s)^*]_{k_1 k_2}ds + [{\bm{\sigma}}(\bm{x}_s^m)d\bm{W}_s(m{\bm{v}}^m_s)^*]_{k_1 k_2} \right ).\nonumber
\end{align}
Now we prove that $\bm{U}_t^m\rightarrow 0$ in $L^2$, and hence in probability, with respect to $C_{\mathbb{R}^d}[0,T]$.  By Lemma~\ref{lemma:convergenceVas}, the first two terms on the right-hand side of equation~(\ref{eq:fullU}) go to zero in $L^2$ with respect to $C_{\mathbb{R}^d}[0,T]$. The rest of the terms in $\bm{U}^m$ are Lebesgue or It\^o integrals with integrands that are products of continuous functions and $m(v_t^m)_i$. {We need a lemma} about the convergence of these integrals to zero.  Recall that in Lemma 2 we have shown that  $m|\bm{v}_t^m| \rightarrow 0$ in $L^2$.  The next lemma proves an explicit bound on the rate of this convergence. 

\begin{lemma} 
\label{lemma:convergenceVpointwise}
For each $m>0$, let $(\bm{x}_t^m, \bm{v}_t^m)$ be the solution to the system ~(1) with functions $\bm{F}$, $\bm{\gamma}$, and $\bm{\sigma}$ satisfying Assumptions 1-3. Then for any fixed $t\in[0,T]$, 
\begin{equation}
	E\left [ m|\bm{v}_t^m|^2\right ] \leq C, 
\end{equation}
where $C$ is a constant independent of $m$ and of $\; t \leq T$. Furthermore, this implies that
\begin{equation}
 E\left [|m\bm{v}_t^m|^2\right ] \leq Cm.
\end{equation}
\end{lemma}

\begin{proof}
Consider the generator of the diffusion process defined by the system (1):
\begin{equation}
	\mathcal{L} = \frac{\sigma_{ik}(\bm{x})\sigma_{jk}(\bm{x})}{2m^2}\frac{\partial^2}
	{\partial v_i \partial v_j} + v_i \frac{\partial}{\partial x_i} + \frac{F_i(\bm{x})}{m}\frac{\partial}{\partial v_i} - \frac{\gamma_{ik}(\bm{x})v_k}{m}\frac{\partial}{\partial v_i} ,
\end{equation}
and apply it to the kinetic energy 
\begin{equation}
	\label{eq:kineticE}
	\phi(\bm{x},\bm{v}) = \frac{m}{2}|\bm{v}|^2. 
\end{equation}
The result is 
\begin{equation}
	\mathcal{L}\phi = \frac{Tr(\bm{\sigma}(\bm{x})\bm{\sigma}^*(\bm{x}))}{2m}+ F_i(\bm{x})v_i
	-\gamma_{ik}(\bm{x})v_kv_i.
\end{equation}
Next, from Assumption~1 we have
\begin{equation}
	\gamma_{ik}(\bm{x})v_kv_i \geq c_\lambda |\bm{v}|^2.
\end{equation}
We use this fact along with the bound
\begin{equation}
	F_{i}(\bm{x})v_i = \left( \frac{F_{i}(\bm{x})}{\sqrt{c_{\lambda}}} \right) (\sqrt{c_{\lambda}}v_i) \leq \frac{1}{2c_{\lambda}}|\bm{F}(\bm{x})|^2 + \frac{c_{\lambda}}{2}|\bm{v}|^2 ,
\end{equation}
to obtain
 \begin{equation}
	\mathcal{L}\phi \leq - \frac{c_{\lambda}}{2} |\bm{v}|^2 + \frac{1}{2c_{\lambda}}|\bm{F}(\bm{x})|^2 +  \frac{Tr(\bm{\sigma}(\bm{x})\bm{\sigma}^*(\bm{x}))}{2m},
\end{equation}
for all $\bm{x} \in \mathcal{U}, \bm{v}\in\mathbb{R}^d$. Recall that for $0 \leq t \leq T$, $\bm{x}^m_t$ lies in the compact set $\mathcal{K}$, so that $|\bm{F}(\bm{x})|$ and $|\bm{\sigma}(\bm{x})|$ are bounded by $C_T>0$ (Assumption 3). Thus, we obtain the bound 
\begin{equation}
	\mathcal{L}\phi(\bm{v}) \leq -\frac{c_{\lambda}}{m}\phi(\bm{v}) + \frac{C_T^2}{2c_{\lambda}}+ \frac{C_T^2d}{2m}
\end{equation}
 For $m<c_{\lambda}d$, the second term is less than the third and thus
\begin{equation}
	\label{eq:MainEst}
	\mathcal{L}\phi(\bm{v}) \leq -\frac{c_{\lambda}}{m}\phi(\bm{v}) +\frac{C_T^2d}{m}.
\end{equation}

Applying the It\^o formula to the process $y_t ^m \equiv \exp(\frac{c_{\lambda}}{m}t)(\phi(\bm{v}^m _t)-\frac{C_T^2d}{c_{\lambda}})$ we obtain
\begin{equation}
	dy ^m _t = 
\left[ \frac{c_{\lambda}}{m}e^{\frac{c_{\lambda}}{m}t}\left(\phi(\bm{v}^m _t)-\frac{C_T^2d}{c_{\lambda}}\right ) +e^{\frac{c_{\lambda}}{m}t}\mathcal{L}\phi(\bm{v} ^m _t) \right] dt + e^{\frac{c_{\lambda}}{m}t} (\bm{v}_t^m)^* \bm{\sigma}(\bm{x}_t^m)\;d\bm{W}_t. 
\end{equation}
Using inequality~(\ref{eq:MainEst}) we obtain,
\begin{equation}
	\frac{c_{\lambda}}{m}e^{\frac{c_{\lambda}}{m}t}\left(\phi(\bm{v}^m _t)-\frac{C_T^2d}{c_{\lambda}}\right ) +e^{\frac{c_{\lambda}}{m}t}\mathcal{L}\phi(\bm{v} ^m _t)
\leq 0.
\end{equation}
Thus, by Dynkin's formula \cite{oksendal}, 
\begin{equation}
	E\left [e^{\frac{c_{\lambda}}{m}t}\left(\phi(\bm{v} ^m _t)-\frac{C_T^2d}{c_{\lambda}}\right )\right ] \leq \frac{m}{2} |\bm{v}|^2 -\frac{C_T^2d}{c_{\lambda}}.
\end{equation}
This implies
\begin{equation}
	\label{eq:mv2}
	E\left [\frac{m}{2}|\bm{v} ^m _t|^2\right ] \leq \frac{C_T^2d}{c_{\lambda}} \left( 1 - e^{-\frac{c_{\lambda}}{m}t} \right ) + \frac{me^{-\frac{c_{\lambda}}{m}t}}{2} |\bm{v}|^2 \leq \frac{C_T^2d}{c_{\lambda}} + \frac{m}{2} |\bm{v}|^2 \leq \frac{C}{2},
\end{equation}
for $C$ independent of $m$.
Q.E.D.
\end{proof}
Now we can prove a lemma to show the integrals in $\bm{U}^m$ converge to zero. 
\begin{lemma}\label{lemma:intmv}
For each $m > 0$, let $\bm{x}_t^m$ be an $\mathcal{F}_t$-adapted process with values in the compact set $\mathcal{K}\subset\mathcal{U}$ for $t \in [0,T]$. If ${g}(x):\mathcal{K}\rightarrow\mathbb{R}$ is a continuous function such that $|g(\bm{x})|\leq C_T$, then for all $\bm{x}\in \mathcal{K}$
\begin{equation}
\lim_{m\rightarrow 0} E\left [ \left (\sup_{0\leq t\leq T} \left |\int_0^t g(\bm{x}_s^m)m(v_s^m)_i\,ds\right |\right )^2\right ] = 0
\end{equation}
and 
\begin{equation}\label{eq:Itomv}
\lim_{m\rightarrow 0} E\left [ \left (\sup_{0\leq t\leq T} \left |\int_0^t g(\bm{x}_s^m)m(v_s^m)_i\,d(W_s)_j\right |\right )^2\right ] = 0,
\end{equation}
for $i=1,...,d, \; j=1,...,k$. 
\end{lemma}
\begin{proof}
First note that, 
\begin{equation}
E\left [ \left (\sup_{0\leq t\leq T} \left |\int_0^t g(\bm{x}_s^m)m(v_s^m)_i\,ds \right | \right )^2 \right ] \leq E\left [\left(\int_0^T \left |g(\bm{x}_s^m)m(v_s^m)_i\right |\,ds \right ) ^2 \right ].
\end{equation}	
 By the Cauchy-Schwarz inequality, 
\begin{align}
E\left [\left (\int_0^T \left |g(\bm{x}_s^m)m(v_s^m)_i\right |\,ds \right )^2 \right ] \leq & T\int_0^T E\left [ ( g(\bm{x}_s^m)m(v_s^m)_i )^2 \right ]\,ds \\
\leq & C_T^2T \int_0^T E\left [(m(v_s^m)_i)^2\right ]\,ds, \nonumber
\end{align}
where the continuous function $g$ is bounded by $C_T$ on $\mathcal{K}$. From  Lemma~\ref{lemma:convergenceVpointwise} we have, 
\begin{equation}
E\left [\left (\int_0^T \left |g(\bm{x}_s^m)m(v_s^m)_i\right |\,ds \right )^2 \right ]\leq T^2 Cm.
\end{equation}
Taking the limit of both sides as $m\rightarrow 0$,
\begin{equation}
\lim_{m\rightarrow 0}E\left [\left (\int_0^T |g(\bm{x}_s^m)m(v_s^m)_i|\,ds \right )^2 \right ] = 0. 
\end{equation}
Therefore,
\begin{equation}
\lim_{m\rightarrow 0}  E\left [ \left (\sup_{0\leq t\leq T} \left |\int_0^t g(\bm{x}_s^m)m(v_s^m)_i\,ds\right |\right )^2\right ] 
\le \lim_{m\rightarrow 0} E\left [\left (\int_0^T |g(\bm{x}_s^m)m(v_s^m)_i|\,ds \right )^2 \right ] 
= 0.
\end{equation}
To estimate the It\^o integral in (\ref{eq:Itomv}), we first use It\^o isometry: 
\begin{align}
E\left [\left(\int_0^T g(\bm{x}_s^m)m(v_s^m)_i\,d(W_s)_j \right)^2 \right ] = & \int_0^T E\left [ (g(\bm{x}_s^m)m(v_s^m)_i )^2\right ]\,ds \\
\leq & C_T^2 \int_0^T E[ (m(v_s^m)_i)^2]\,ds. \nonumber
\end{align}
Using Lebesgue dominated convergence theorem and Doob's maximal inequality (see page 14 of \cite{karatzas}), 
\begin{equation}
E\left [ \left (\sup_{0\leq t\leq T} \left |\int_0^t g(\bm{x}_s^m)m(v_s^m)_i\,d(W_s)_j\right |\right )^2\right ] \leq 4 E\left [\left( \int_0^T g(\bm{x}_s^m)m(v_s^m)_i\,d(W_s)_j \right)^2 \right ] \rightarrow 0
\end{equation}
as $m\rightarrow 0$. 

Q.E.D.
\end{proof}

We use Lemma~\ref{lemma:intmv} to show $\bm{U}^m$ converges to zero in $L^2$ with respect to $C_{\mathbb{R}^d}[0,T]$ as $m\rightarrow0$. Note that all functions in the expression~(\ref{eq:fullU}) for $\bm{U}^m$ are continuous. The integrals $\int_0^\infty(e^{-\bm{\gamma}(\bm{x}_s^m)y})_{jk_1}(e^{-\bm{\gamma}^*(\bm{x}_s^m)y})_{k_2 l}\,dy$ are continuous because $\bm{\gamma}$ is continuous, matrix exponentiation is a continuous operation and the integrand decays exponentially with $y$. Therefore, $\bm{U}^m\rightarrow 0$ as $m\rightarrow 0$ in $L^2$ with respect to $C_{\mathbb{R}^d}[0,T]$. 

To verify the rest of the assumptions of Lemma~\ref{theorem:KP}, including Condition~\ref{condition}, we first write equation~(\ref{eq:xmcomplete}) in the form
\begin{equation}
\bm{x}_t^m = \bm{x} + \bm{U}_t^m + \int_0^t \bm{f}(\bm{x}_t^m)\,d\bm{H}_t^m.
\end{equation}
Define $\bm{f}:\mathcal{U}\rightarrow \mathbb{R}^{d\times (1 + k + 1 + d^2)}$  as 
\begin{equation}
\bm{f}(\bm{x}) = \begin{pmatrix} \bm{\gamma}^{-1}(\bm{x})\bm{F}(\bm{x}), & \bm{\gamma}^{-1}(\bm{x})\bm{\sigma}(\bm{x}), & \bm{S}(\bm{x}), & \bm{f}^1(\bm{x}), ..., \bm{f}^{d}(\bm{x}) \end{pmatrix}
\end{equation}
where the components of $\bm{S}(\bm{x}):\mathcal{U}\rightarrow \mathbb{R}^d$ are defined as
\begin{equation}
{S}_i(\bm{x}) =  \int_0^t\frac{\partial}{\partial x_{l}}[(\gamma^{-1})_{ij}(\bm{x}^m_s)]J_{j{l}}(\bm{x}),
\end{equation}
$\bm{J}$ is the solution of the Lyapunov equation~(\ref{eq:Glyapunov}) and the components of $\bm{f}^{\beta}(\bm{x}):\mathcal{U}\rightarrow \mathbb{R}^{d\times d}$ are defined as
\begin{equation}
f^{k_2}_{i k_1}(\bm{x}) =\int_0^t\frac{\partial}{\partial x_{l}}[(\gamma^{-1})_{ij}(\bm{x})]
\left [- \int_0^\infty (e^{-\bm{\gamma}(\bm{x})y})_{jk_1}(e^{-\bm{\gamma}^*(\bm{x})y})_{k_2 l}\,dy \right ]
\end{equation}
for $k_1 ,k_2=1,2,...,d$. Next, $\bm{H}^m_t$ with paths in $C_{\mathbb{R}^{1+k+1+d^2}}[0,T]$ is defined as
\begin{equation}
\bm{H}^m_t = \begin{pmatrix} t \\ \bm{W}_t \\ t \\ (mv_t^m)_1m\bm{v}_t^m-mv_1m\bm{v} \\ \vdots \\ (mv_t^m)_d m\bm{v}_t^m-mv_d m\bm{v} \end{pmatrix}.
\end{equation}
By Lemma~\ref{lemma:convergenceVas}, $\bm{H}^m\rightarrow \bm{H}$ as $m\rightarrow 0$ in probability with respect to $C_{\mathbb{R}^{1+k+1+d^2}}[0,T]$, where 
\begin{equation}
\bm{H}_t = \begin{pmatrix} t \\ \bm{W}_t \\ t \\ 0 \\ \vdots \\ 0 \end{pmatrix}.
\end{equation}
Therefore, $(\bm{U}^m,\bm{H}^m)\rightarrow (\bm{0},\bm{H})$ as $m\rightarrow 0$ in probability with respect to $C_{\mathbb{R}^d\times\mathbb{R}^{1+k+1+d^2}}[0,T]$. All that is left, to be able to use Lemma~\ref{theorem:KP}, is to check Condition~\ref{condition}.

\subsection{Verification of Condition~1}
We need to find the Doob-Meyer decomposition of $\bm{H}_t^m$ and stochastically bound, uniformly in $m$, the bounded variation part of the decomposition, denoted $\bm{A}_t^m$. Only the last  $d^2$ rows of $\bm{H}^m$ depend on $m$. Furthermore, the columns of the matrix $(m\bm{v}_t^m (m\bm{v}_t^m)^*)$ make up the last $d^2$ rows of $\bm{H}^m$. That is, the first column of the matrix $(m\bm{v}_t^m (m\bm{v}_t^m)^*)$ is rows $1+k+1+1$ through $1+k+1+d$ of $\bm{H}^m$. The second column of the matrix $(m\bm{v}_t^m (m\bm{v}_t^m)^*)$ is rows $1+k+1+d+1$ through  $1+k+1+2d$ of $\bm{H}^m$ and so on. Consider the expression for $d(m\bm{v}_t^m(m\bm{v}_t^m)^*)$ given by equation~(\ref{eq:dmvmv}). Because the stochastic integrals are local martingales, $\bm{A}_t^m$ contains the columns of the Lebesgue integrals in the above expression. That is, 
\begin{equation}
\bm{A}_t^m = \begin{pmatrix} t \\ 0 \\ t \\ (\bm{\mathcal{A}}_t^m)^1 \\ \vdots \\ (\bm{\mathcal{A}}_t^m)^d \end{pmatrix},
\end{equation}
where
\begin{align}\label{eq:Atm}
\begin{pmatrix} (\bm{\mathcal{A}}_t^m)^1 ,& (\bm{\mathcal{A}}_t^m)^2, & \cdots, & (\bm{\mathcal{A}}_t^m)^d \end{pmatrix} =& \int_0^t m\bm{v}_s^m\bm{F}(\bm{x}_s^m)^*\,ds \\
+& \int_0^t\bm{F}(\bm{x}_s^m)(m\bm{v}_s^m)^*ds \nonumber\\
-&  \int_0^tm(\bm{v}_s^m) (\bm{\gamma}(\bm{x}_s^m)\bm{v}_s^m)^*\,ds \nonumber\\
-& \int_0^t\bm{\gamma}(\bm{x}_s^m)\bm{v}_s^m m(\bm{v}_s^m)^* \,ds \nonumber \\
+& \int_0^t\bm{\sigma}(\bm{x}_s^m)\bm{\sigma}^*(\bm{x}_s^m)\,ds.\nonumber
\end{align}  
We must show that $\bm{A}_t^m$ is stochastically bounded. Because $m\bm{v}^m\rightarrow 0$ in probability, the first and second terms on the right-hand side of equation~(\ref{eq:Atm}) go to zero in probability. By Assumption~3, $\bm{\sigma}(\bm{x}_t^m)\bm{\sigma}^*(\bm{x}_t^m)$ is bounded for all $t\in[0,T]$ and thus the fifth term is stochastically bounded in $m$. To prove stochastic boundedness of the third and fourth terms, it is enough to show $E[|m(v_s^m)_i(\bm{v}_s^m)|]$ is bounded uniformly in $m$ (based on previous works \cite{pavliotis,kupferman2004,hottovy2012}, we expect $\sqrt{m}\bm{v}_s^m$ to be of order one). For the rows we have $|m(v_s^m)_i(\bm{v}_s^m)|\leq m|\bm{v}_s^m|^2$ for every $i=1,...,d$. Using 
Lemma~\ref{lemma:convergenceVpointwise} we have
\begin{equation}
	E[ m|\bm{v}_s^m|^2] \leq C,
\end{equation}
uniformly in $m$.
 Thus, by the Chebyshev inequality, $\{V_t(\bm{A}^m)\}$ is stochastically bounded and this proves that $\bm{H}_t^m$ satisfies Condition~\ref{condition}. 

Therefore, $\bm{x}_t^m\rightarrow\bm{x}_t$ in probability as $m\rightarrow 0$. We use this together with boundedness to prove $L^2$ convergence: because $\bm{x}_t^m$ lies in a bounded set $\mathcal{K}$, there exists $N>0$ such that $P(|\bm{x}_t^m|\leq N)=1$ for all $t$ and $m$. Therefore, 
\begin{align}
	\lim_{m\rightarrow0}E\left [\left (\sup_{0\leq t\leq T}|\bm{x}_t^m-\bm{x}_t|\right )^2\right ] =& \lim_{m\rightarrow 0}\int_0^\infty P\left[ \left (\sup_{0\leq t\leq T}|\bm{x}_t^m-\bm{x}_t|\right )^2\geq x \right] \,dx \\
=& \int_0^{(2N)^2}\lim_{m\rightarrow0}P\left[ \left (\sup_{0\leq t\leq T}|\bm{x}_t^m-\bm{x}_t|\right )^2\geq x \right] \,dx \nonumber \\
=& 0. \nonumber
\end{align}
Q.E.D.
\end{proof}

\begin{remark}\label{remark2}
One may be tempted to apply Lemma~\ref{theorem:KP} to equation (\ref{eq:sub}) without integration by parts, because $m\bm{v}_t^m\rightarrow 0$. However, this would lead to the limiting equation,
\begin{equation}
	\label{eq:wronglimit}
	d\bm{x}_t = \bm{\gamma}^{-1}(\bm{x}_t)\bm{F}(\bm{x}_t) \,dt + \bm{\gamma}^{-1}(\bm{x}_t)\bm{\sigma}(\bm{x}_t)\,d\bm{W}_t. 
\end{equation}
This is not the equation we derived. In view of Lemma~\ref{lemma:convergenceVas}, if $\bm{\gamma}(\bm{x})=\bm{\gamma}_0$ is a constant matrix for all $\bm{x}$, then 
\begin{equation}
\lim_{m\rightarrow 0}P\left (\left (\sup_{0\leq t\leq T}\left |\int_0^tm\bm{\gamma}_0^{-1}d\bm{v}_s^m 
\right |\right )^2>\epsilon\right ) = \lim_{m\rightarrow 0}P\left (\left (\sup_{0\leq t\leq T}\left |\bm{\gamma}_0^{-1}m\bm{v}_t^m-\bm{\gamma}_0^{-1}m\bm{v}\right |\right )^2>\epsilon\right ) = 0,
\end{equation}
 similarly to \cite{nelson,freidlin2004}.  However, with $\bm{\gamma}(\bm{x})$ dependent on position, the limit will be non-zero because $m\bm{v}_t^m$ does not satisfy Condition~\ref{condition}.  Note that from the SDE~(\ref{eq:SDEgeneral}) for $d\bm{v}_t^m$
\begin{equation}
	m\bm{v}_t^m = m\bm{v} + \underbrace{\int_0^t\left (\bm{F}(\bm{x}_s^m)-\bm{\gamma}(\bm{x}_s^m)\bm{v}_s^m\right )ds}_{\bm{A}_t^m \text{ Bounded Variation}} \; + \underbrace{\int_0^t\bm{\sigma}(\bm{x}_s^m)\,d\bm{W}_s}_{\bm{M}_t^m \text{ Local Martingale}}. 
\end{equation}
 Because the limits of integration are finite, $\bm{A}_t^m$ has bounded variation for fixed $m>0$.  Note that $O(V_t(\bm{A}^m)) = O(\bm{v} _t ^m )$.  It can be shown explicitly in the special case in which the fluctuation-dissipation relation is satisfied (and we expect it to be true in general) that $\bm{v} _t ^m$ is of the order $m^{- \frac{1}{2}}$.  Therefore $O(V_t(\bm{A}^m)) = O(m^{-1/2})$ and Lemma~\ref{theorem:KP} cannot be used. 
\end{remark}

\section{One Dimension}\label{sec:1D}

As the first example, we apply Theorem~\ref{theorem} to a one-dimensional model of a Brownian particle. This is the model studied in \cite{hottovy2012} and earlier in \cite{sancho1982}. The particle's position satisfies
\begin{equation}\label{eq:SDEgeneral1D}
	\left \{ \begin{array}{rcl}
            dx_t^m &=& v_t^m\,dt\\
	dv_t^m &=& \left ( \frac{{F}(x_t^m)}{m} - \frac{{\gamma}(x_t^m)}{m}v_t^m\right )\,dt + \frac{{\sigma}(x_t^m)}{m}dW_t
	\end{array} \right .
\end{equation}
with initial conditions $x_0^m = x$ and $v_0^m = v$. For simplicity, we study the system on the whole real line, assuming the coefficients and their derivatives are bounded. These assumptions will be relaxed in Section~\ref{sec:BP}. Equation~(\ref{eq:Glyapunov}) for the noise-induced drift term is in this case
\begin{equation}
	2J(x)\gamma(x) = \sigma(x)^2.
\end{equation}
Thus, the limiting equation for $x_t$ is
\begin{equation}
	\label{eq:SKlimit1D}
	dx_t = \left (\frac{{F}(x_t)}{{\gamma}(x_t)} - \frac{{\gamma}'(x_t)}{2{\gamma}(x_t)^3}{\sigma}(x_t)^2\right )dt + \frac{{\sigma}(x_t)}{{\gamma}(x_t)}dW_t,
\end{equation}
with $x_0 = x$, which recovers prior results \cite{sancho1982,freidlin2011,hottovy2012}. 

It is instructive to illustrate on this simple example the key quantities entering the proof of Lemma~1, namely $\bm{f}$ and $\bm{H}_t^m$. Define $\bm{f}$, a continuous function from $\mathbb{R}$ to $\mathbb{R}^4$, as 
\begin{equation}
	\bm{f}(x) = \begin{pmatrix} \frac{{F}(x)}{{\gamma}(x)}, &\frac{{\sigma}(x)}{{\gamma}(x)}, &  - \frac{{\gamma}'(x)}{2{\gamma}(x)^3}{\sigma}(x)^2, &  - \frac{{\gamma}'(x)}{{\gamma}(x)^3} \end{pmatrix},
\end{equation}
and $\bm{H}_t^m$ with paths in $C_{\mathbb{R}^4}[0,T]$ as, 
\begin{equation}
	\bm{H}_t^m = \begin{pmatrix} t \\ W_t \\ t \\ \frac{1}{2}\left [(mv_t^m)^2 - (mv)^2\right ] \end{pmatrix}.
\end{equation}
We have $\lim_{m\rightarrow 0} \bm{H}_t^m = ( t, W_t, t, 0)^*$, and the limiting equation~(\ref{eq:SKlimit1D}) is recovered. 

The boundedness of the coefficients and their derivatives implies global existence of the strongly unique solutions 
$x_t^m$ to SDE~(\ref{eq:SDEgeneral}) for every $m>0$, and $x_t$ to SDE~(\ref{eq:SKlimit}); Assumptions 1-3 are thus satisfied. {However, because the
state space of the process (the real line) is unbounded, we can only conclude convergence  in probability 
(for comparison, see the last paragraph of the proof of Theorem~\ref{theorem}). Therefore, by Theorem~\ref{theorem}, $x_t^m\rightarrow x_t$ as $m\rightarrow 0$ in probability with respect to $C_{\mathbb{R}}[0,T]$.}

\section{Examples and applications}\label{sec5}

\subsection{Brownian particle in a {one-dimensional} diffusion gradient}\label{sec:BP}
The equations studied in this example model the experiment described in \cite{volpe2011}. In this experiment a colloidal particle is diffusing in a cylinder filled with water. The friction and noise coefficients depend on the particle'��s position, as described below, giving rise to a noise-induced drift. Even though we do not verify Assumption~2 in this case, the Smoluchowski-Kramers approximation derived in Theorem~\ref{theorem} agrees with the experimental results of \cite{volpe2011}. The equations are:
\begin{equation}\label{eq:LE}
\left\{\begin{array}{ccl}
	dx_t^m &=& v_t^m\;dt \\
	dv_t^m &=& \left [\frac{F(x_t^m)}{m} - \frac{k_BT}{mD(x_t^m)}v_t^m\right ]\;dt + 
	\frac{k_BT\sqrt{2}}{m\sqrt{D(x_t^m)}}\;dW_t
\end{array}\right.
\end{equation}
where $D(x)$ is the diffusion coefficient. Near $x = 0$ $D(x)$ can be expressed analytically \cite{brenner} and has the form shown in Fig.~1; an analogous behavior also holds near the top of the cylinder. The force $F$ results from effective gravity and electrostatic repulsion from the bottom and top walls of the container. Away from the lateral walls of the cylinder both forces are vertical so the horizontal components of particle's motion can be (and were) separated and the equations are written for the vertical component only.
\begin{figure}[h!]
\resizebox{.5\textwidth}{!}{
\includegraphics{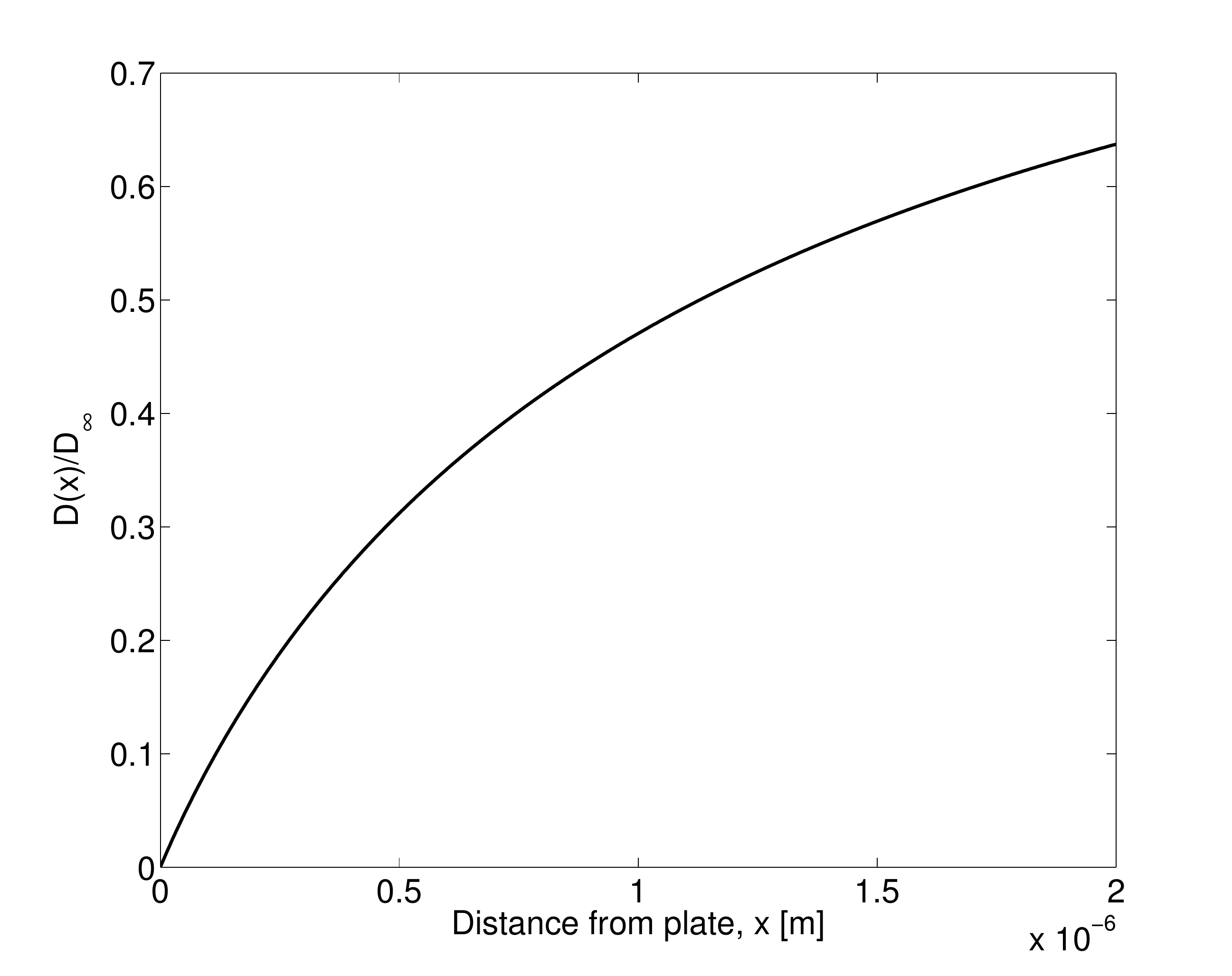}
}
\caption{Plot of the normalized diffusion coefficient $D(x)$ for a spherical particle of radius $1\,\rm{\mu m}$.}
\label{fig:D}
\end{figure}

An application of equation~(\ref{eq:SKlimit1D}) to this case gives the limiting equation
\begin{equation}
	dx_t = \left [\frac{D(x_t)F(x_t)}{k_BT} + D'(x_t)\right ]\;dt + \sqrt{2D(x_t)}\;dW_t. 
\end{equation}
The noise-induced term in the drift is thus $S(x) = D'(x)$, as observed in \cite{volpe2011}.

\subsection{Systems driven by a colored noise}\label{sec:OU}
The driving mechanisms of real physical systems are typically characterized by a non-zero correlation time. Therefore, models employing colored noise, instead of white noise, are often more appropriate to describe them. We work through two examples with Ornstein-Uhlenbeck colored noise. We calculate the limiting equations without stating explicit conditions for the existence and uniqueness assumed in Theorem~\ref{theorem}. In this Section we consider the multi-dimensional version of  equation~(\ref{eq:SDEOUCN}):
\begin{equation}\label{eq830000}
	\left \{\begin{array}{rcl}
	d\bm{x}_t &=& \bm{v}_t\,dt \\
	d\bm{v}_t &=&\left[ \frac{ \bm{F}(\bm{x}_t)}{m} - \frac{\bm{\gamma}(\bm{x}_t)}{m}\bm{v}_t + \frac{\bm{\sigma}(\bm{x}_t)}{m}\bm{\eta}_t \right] dt
	\end{array}\right .
\end{equation}
where $\bm{x}_t\in\mathcal{U}\subset\mathbb{R}^d$ and $\bm{\eta}_t$ is a $k$-dimensional stationary random process with zero mean and correlation time $\tau$.  To use the framework of Theorem~\ref{theorem}, we consider a special type of noise, the Ornstein-Uhlenbeck process defined as the stationary solution of the SDE
\begin{equation}
\label{eq:OU}
	d\bm{\eta}_t = -\frac{\bm{A}}{\tau}\bm{\eta}_t\,dt + \frac{\bm{\lambda}}{\tau}d\bm{W}_t, 
\end{equation}
where $\bm{A}$ is a $k$ by $k$ constant invertible matrix, $\bm{\lambda}$ is a $k$ by $\ell$ constant matrix, and $\bm{W}$ an $\ell$-dimensional Wiener process. Defining the variable $\bm{\zeta}_t$ by the equation $d\bm{\zeta}_t=\bm{\eta}_t\,dt$, we use the above framework by setting $\bar{\bm{x}} = (\bm{x},\bm{\zeta})$ and $\bar{\bm{v}} = (\bm{v},\bm{\eta})$.  We will now illustrate this use of Theorem~\ref{theorem} to derive the limit, as the correlation time $\tau$ and mass $m$ tend to zero, on two concrete examples.  Note that here the initial condition $\bm{\eta}_0$ is taken to be a random variable distributed according to the stationary distribution corresponding to (\ref{eq:OU}), so that it is Gaussian and depends on $\tau$, but this presents no additional difficulty and the theorem can be generalized to include this case.

\subsubsection{A system with colored noise and constant friction}\label{sec:constfric}
Consider the system 
\begin{equation}
\left \{ \begin{array}{rcl}
	\mu \ddot{x}_t &= &{F}(x_t) + \left[ -\dot{x}_t + f(x_t)\eta_t \right]\\
	d{\eta}_t &=& - \frac{a \eta_t}{\epsilon^2} \,dt + \frac{\sqrt{2\lambda}}{\epsilon^2}\,dW_t
\end{array} \right.
\end{equation}
with $x_t$ and $\eta_t$ one-dimensional.
This is equivalent to the example in \cite[Section 11.7.6]{pavliotis} with the substitution $\eta_t = \frac{1}{\epsilon}\tilde{\eta_t}$, where $\tilde{\eta}_t$ is the colored noise used in the reference. Setting $\mu = k \epsilon^2$, we rewrite the above system as
\begin{equation}
\label{eq:OUconstFric}
\left \{ \begin{array}{rcl}
	dx_t &=& v_t\,dt \\
	dv_t &=& \left[ \frac{F(x_t)}{k \epsilon^2}  -\frac{v_t}{k \epsilon^2} + \frac{f(x_t)\eta_t}{k \epsilon^2} \right] dt \\
	d\zeta_t &=&  \eta_t dt \\
	d\eta_t &=& -\frac{a \eta_t}{\epsilon^2}\,dt + \frac{\sqrt{2\lambda}}{\epsilon^2}\,dW_t 
\end{array} \right .
\end{equation}
In the framework of Section~\ref{sec:SKa}, defining $\bm{x}_t = (x_t,\zeta_t)^*$ and $\bm{v}_t = (v_t,\eta_t)^*$, and letting $m=\epsilon^2$, the SDE system~(\ref{eq:OUconstFric}) becomes
\begin{equation}
\left \{ \begin{array}{rcl}
d\bm{x}_t & = & \bm{v}_t \,dt \\
md\bm{v}_t &=& \tilde{\bm{F}}(\bm{x}_t)dt - \bm{\gamma}(\bm{x}_t)\bm{v}_t\,dt + \bm{\sigma}(\bm{x}_t)d{W}_t
\end{array}\right .
\end{equation}
with 
\begin{equation}
\tilde{\bm{F}}(\bm{x}_t) = \begin{pmatrix} \frac{{F}(x_t)}{k} \\ 0 \end{pmatrix}, \quad \bm{\gamma}(\bm{x}_t) = \begin{pmatrix} \frac{1}{k} & -\frac{f(x_t)}{k} \\ 0 & a \end{pmatrix},\quad \bm{\sigma}(\bm{x}_t) = \begin{pmatrix} 0 \\ \sqrt{2\lambda} \end{pmatrix}. 
\end{equation} 
To compute the noise-induced drift term, we solve the Lyapunov equation, 
\begin{equation}
	\bm{\gamma}\bm{J} +\bm{J}\bm{\gamma}^* = \bm{\sigma}\bm{\sigma}^*, 
\end{equation}
and note that the Wiener process $W_t$ is one-dimensional. 
We use Mathematica\textsuperscript{\textregistered} to find a closed form for $\bm{J}$, 
\begin{equation}
	\bm{J}(\bm{x}) = \begin{pmatrix} \frac{\lambda f(x)^2}{a(1+a k)} &  \frac{\lambda f(x)}{a(1+a k)} \\
		\frac{\lambda f(x)}{a(1+a k)} & \frac{\lambda}{a} \end{pmatrix}.
\end{equation}
We compute the noise-induced drift in the first component ($i = 1$) using equation~(\ref{eq:spurious}):
\begin{equation}
	\begin{array}{rcl}
	S_1(x) &=&\frac{\partial}{\partial x_{l}}[({\gamma}^{-1})_{1j}({x})]J_{j{l}}(x)\\
	& = & \frac{\lambda f'(x)f(x)}{a^2(1+ k a)}.
	\end{array}
\end{equation}
Therefore, the limiting SDE for $x_t$ is
\begin{equation}
	dx_t =\left[ {F}(x_t) + \frac{\lambda f'(x_t)f(x_t)}{a^2(1+ k a)} \right]dt + \sqrt{\frac{2\lambda}{a^2}}f(x_t)\,dW_t,
\end{equation}
in agreement with \cite{pavliotis}. 

\subsubsection{Thermophoresis}
\label{sec:thermo}

The same type of equation can be used to model thermophoresis, i.e. the movement of small particles in a temperature gradient \cite{piazza2008}. While theoretical models of this phenomenon are still a matter of debate, thermophoresis has been successfully employed experimentally, e.g., to separate and group small particles \cite{piazza2008} and to influence the motion of DNA \cite{duhr2006pnas}. In \cite{hottovyEPL2012} we used equation~(\ref{eq830000}) to model the motion of a particle of mass $m$ driven by a colored noise $\eta_t$ with a short correlation time $\tau$ in an environment where the temperature $T(x)$ depends on the particle's position $x$, and thus $\gamma(x) = \gamma(T(x))$ and $D(x) = D(T(x))$. In the limit as $m,\tau\rightarrow 0$, the noise-induced drift pushes the particle toward the hotter regions or toward the colder regions depending on the ratio $m/\tau$. This was argued in \cite{hottovyEPL2012} using a multi-scale expansion. We now show this using Theorem~\ref{theorem}. We consider the SDE system
\begin{equation}\label{eq:thermotheta}
	\left \{\begin{array}{rcl}
	dx_t &=& v_t\,dt \\
	dv_t &=&\left[ \frac{F(x_t)}{\theta(x_t) \tau} -\frac{1}{\theta(x_t) \tau}v_t + \frac{\sqrt{2D(x_t)}\eta_t}{\theta(x_t) \tau}\right] dt \\
	d\zeta_t &=& \eta_t\,dt \\
	d\eta_t &=& -\frac{2\eta_t}{\tau}\,dt + \frac{2}{\tau}\,dW_t
	\end{array}\right .
\end{equation}
where $W_t$ is a one-dimensional Wiener process and we have introduced the dimensionless quantity
\begin{equation}
	\theta(x) = \theta(T(x)) = \frac{m}{\bm{\gamma}(T(x))\tau}.
\end{equation}
Differently from previous sections, the small parameter is $\tau$, not $m$ (as $\tau$ goes to zero $m$ will go to zero as well). Define $\bm{x} = (x,\zeta)$, $\bm{v} = (v,\eta)$, and 
\begin{equation}
	\bm{\gamma}(\bm{x}) = \begin{pmatrix} \frac{1}{\theta(x)} & -\frac{\sqrt{2D(x)}}{\theta(x)} \\
0 & 2 \end{pmatrix}, \quad \bm{\sigma} = \begin{pmatrix} 0 \\ 2 \end{pmatrix}.
\end{equation}
$\bm{\gamma}$ is invertible and
\begin{equation}
	\bm{\gamma}^{-1}(\bm{x}) = \begin{pmatrix} \theta(x) & \frac{\sqrt{2D(x)}}{2} \\ 0 & \frac{1}{2} \end{pmatrix}.
\end{equation}
To compute the noise-induced drift term, we solve the Lyapunov equation, 
\begin{equation}
	\bm{\gamma}\bm{J} +\bm{J}\bm{\gamma}^* = \bm{\sigma}\bm{\sigma}. 
\end{equation}
A closed form of $\bm{J}$ obtained using Mathematica\textsuperscript{\textregistered} is
\begin{equation}
	\bm{J}(\bm{x}) = \begin{pmatrix} \frac{2 D(x)}{1+2\theta(x)} & \frac{ \sqrt{2D(x)}}{1+2\theta(x)}\\ \frac{ \sqrt{2D(x)}}{1+2\theta(x)} & 1\end{pmatrix}
\end{equation}
Using equation~(\ref{eq:SKlimit}), as $\tau,m\rightarrow 0$ so that $m/\tau$ is constant, we see that the limiting equation for $x$ is
\begin{equation}
	\label{eq:thermolimit}
	dx_t =\left[ \frac{{F}(x_t)}{\theta({x}_t)}+ \frac{\bm{\gamma}(x_t)D'(x_t)-4\theta(x_t)\bm{\gamma}'(x_t)D(x_t)}{2\bm{\gamma}(x_t)(1+2\theta(x_t))} \right] dt + \sqrt{2D(x_t)}\,dW_t,
\end{equation}
which coincides with the result of \cite{hottovyEPL2012}.  
\begin{remark}
Strictly speaking, the system~(\ref{eq:thermotheta}) does not obey the fluctuation-dissipation relation as the time correlations of the noise should be reflected in the friction term, which should become an integral over the past \cite[Section 1.5]{zwanzig}.  The resulting non-Markovian system requires a more refined analysis.
\end{remark}

\subsection{Three-dimensional Brownian motion in a force field}\label{sec:3DBP}
As a generalization of the example in Section~\ref{sec:BP}, we consider a Brownian particle in  $\mathbb{R}^3$. The coefficients consist of a spatially varying noise coefficient $\bm{\sigma}(\bm{x})$ and the fluctuation-dissipation relation \cite{kubo} in multi-dimensional form, i.e.
\begin{equation}
	\bm{\gamma}(\bm{x}) = \frac{\bm{\sigma}(\bm{x})\bm{\sigma}^*(\bm{x})}{k_BT}. 
\end{equation}
A force $\bm{F}$ is acting on the particle.
Equation~(\ref{eq:SDEgeneral}) becomes
\begin{equation} \label{BMforce}
\left\{\begin{array}{rcl}
d\bm{x}_t^m & = & \bm{v}_t^m\,dt \\
d\bm{v}_t^m & = & \left[ \frac{\bm{F}(\bm{x}_t^m)}{m} - \frac{\bm{\sigma}\bm{\sigma}^*(\bm{x}_t^m)}{mk_BT}\bm{v}^m_t \right] \,dt + \frac{\bm{\sigma}(\bm{x}_t^m)}{m}\,d\bm{W}_t
\end{array}\right.
\end{equation}
To find the limiting equations, we solve the Lyapunov equation
\begin{equation}
	\frac{1}{k_BT}\left(\bm{\sigma}\bm{\sigma}^* \bm{J} + \bm{J}\bm{\sigma}\bm{\sigma}^*\right) =  \bm{\sigma}\bm{\sigma}^* 
\end{equation}
obtaining $\bm{J} = \frac{k_BT}{2}\bm{I}$ where $\bm{I}$ is the identity matrix. The limiting equation~(\ref{eq:SKlimit}), as $m\rightarrow 0$, is
\begin{equation}\label{eq113dhdkodl}
	d\bm{x}_t=\left [ (\bm{\sigma}\bm{\sigma}^*(\bm{x}_t))^{-1}k_BT\bm{F}(\bm{x}_t)-k_BT\bm{S}(\bm{x}_t)\right ]dt +  [\bm{\sigma}(\bm{x}_t)^*]^{-1}k_BTd\bm{W}_t,
\end{equation}
where the $i^\text{th}$ component of $\bm{S}$ equals
\begin{equation}
\label{eq:spuriousG}
S_{i}(\bm{x}) =\frac{k_BT}{2}\frac{\partial}{\partial x_{l}}([(\bm{\sigma}\bm{\sigma}^*)^{-1}(\bm{x})]_{i{l}}).
\end{equation}

\begin{remark}
If $\bm{F}$ is a conservative force, i.e. $ \bm{F} = -\nabla{U}$, it can be shown (e.g. by solving the corresponding stationary Fokker-Planck equation) that for $m > 0$  equation~(\ref{BMforce}) has a stationary density $C\exp\left \{ -\frac{{U}(\bm{x})}{k_BT} - \frac{m|\bm{v}|^2}{2k_BT}\right \}$ (Gibbs distribution). In this case, one can recover the formula for $\bm{S}$ by requiring that the limiting equation has $C\exp\left \{ -\frac{{U}(\bm{x})}{k_BT} \right\}$ as its stationary density. For a non-conservative force $\bm{F}$, the stationary solution will not be Gibbs and the limit is identified using Theorem~\ref{theorem}. Interestingly these cases have also been studied experimentally in the presence, e.g., of non-conservative forces arising from hydrodynamic interactions in two dimensions \cite{volpe2008} and optical forces in three dimensions \cite{simpson1997,pesce2009}.
\end{remark}

\subsection{Brownian particle in a three-dimensional magnetic field}\label{sec:magnetic}
We consider a particle of mass $m$ and charge $q$, moving in three dimensions under an external force ${\bm F}({\bm x})$ and a friction force $-{\bm \gamma}({\bm x}){\bm v}$ in the presence of (white) noise ${\bm \sigma}({\bm x}){\bm \eta_t}$.  We assume there is an additional magnetic (Lorentz) force $q{\bm v} \times {\bm B}({\bm x})$, where ${{\bm B}\in\mathbb{R}^3}$ is a magnetic field. Similar problems were studied in \cite{kwon2005structure,cerrai2011,freidlin2012}. The Lorentz force can be written as an action of an (antisymmetric) matrix ${{\bm H}({\bm x}) \in C_{\mathbb{R}^{3\times3}}[0,T]}$ on ${\bm v}$.  While physically ${\bm H}({\bm x})$ does not represent friction, it can be added to the friction term, changing the matrix ${\bm \gamma}$ to a modified one
$$
\tilde{\bm{\gamma}}({\bm x}) = {\bm \gamma}({\bm x}) +  {\bm H}({\bm x}).
$$
Note that $\bm{\gamma}$ and $\tilde{\bm{\gamma}}$ have the same symmetric part and, therefore, Assumption~\ref{assume:bddcoeffs} is preserved. Accordingly, the noise-induced drift $\tilde{\bm S}$ is now calculated, using the solution of the modified Lyapunov equation
\begin{equation}
	\label{eq:LyapMag}
	\bm{{J}}\bm{\tilde{\gamma}}^* + \bm{\tilde{\gamma}}\bm{{J}} = \bm{\sigma}\bm{\sigma}^*,
\end{equation}
In particular, if ${\bm \gamma}$ and ${\bm \sigma}$ satisfy the Einstein relation ${\bm \sigma}{\bm \sigma}^* = 2k_BT {\bm \gamma}$, the solution of the Lyapunov equation is
$$
{\bm J} = k_BT {\bm I},
$$
where ${\bm I}$ is the identity matrix, leading to
$$
\tilde{S}_i({\bm x}) = {k_BT}{\partial \over \partial x_j}[( \bm{\gamma}+ \bm{H})^{-1}_{ij}({\bm x})].
$$
The result in this case is essentially contained (based on different arguments) in \cite{shi2012}. This case is special in that adding an anti-symmetric matrix $\bm{H}$ to $\bm{\gamma}$  does not change the solution of the Lyapunov equation.

\section{Stratonovich form of the limiting equation}\label{sec4}
In general, an It\^o system 
\begin{equation}
	d(x_t)_i = b_i(\bm{x}_t)\;dt + h_{ij}(\bm{x}_t)\;d (W_t)_j
\end{equation}
has an equivalent Stratonovich form
\begin{equation}
	d(x_t)_i = b_i(\bm{x}_t)\;dt - \frac{1}{2} \left( \partial_k(h_{ij}) (\bm{x}_t) \right) h_{kj}(\bm{x}_t) dt + h_{ij}(\bm{x}_t)\circ d(W_t)_j,
\end{equation}
in which the middle term $- \frac{1}{2} \left( \partial_k(h_{ij}) (\bm{x}_t) \right) h_{kj} (\bm{x}_t)$ is the {\it It\^{o}-to-Stratonovich
correction}. We apply it to equation~(\ref{eq:SKlimit}), where $\bm{h} = \bm{\gamma}^{-1}\bm{\sigma}$, getting for the It\^{o}-to-Stratonovich
correction the expression 
\begin{equation}
	-\frac{1}{2} (\partial_k(\gamma^{-1})_{i\ell}) \sigma_{\ell j}(\gamma^{-1})_{km}\sigma_{mj} - \frac{1}{2} (\gamma^{-1})_{i\ell} (\partial_k(\sigma_{\ell j})) (\gamma^{-1})_{km}\sigma_{mj}. 
\end{equation}
In the case when $\bm{\gamma} = \bm{\gamma}^*$ commutes with $\bm{\sigma}$ (and thus
also with $\bm{\sigma}^*$), the solution of the Lyapunov equation~(\ref{eq:Lyapunov}) is 
\begin{equation}
	\bm{J} = \frac{1}{2}\bm{\sigma}\bm{\sigma}^*\gamma^{-1}. 
\end{equation}
Substituting it into the limiting equation~(\ref{eq:spurious}) we see that $\bm{S}$ cancels the first term of the It\^o-to-Stratonovich correction and thus in the Stratonovich language the limiting equation becomes
\begin{equation}
	\label{eq:StratSKlimit}
	d\bm{x}_t = \left[ \bm{\gamma}^{-1}(\bm{x}_t)\bm{F}(\bm{x}_t)+ \bar{\bm{S}}(\bm{x}_t)
	\right] \;dt + \bm{\gamma}^{-1}(\bm{x}_t)\bm{\sigma}(\bm{x}_t)\circ d\bm{W}_t, 
\end{equation}
with
\begin{equation}
	\bar{S}_i(\bm{x}) = - \frac{1}{2}(\gamma^{-1})_{i\ell} (\bm{x}) ( \partial_k(\sigma_{\ell j})(\bm{x}))(\gamma^{-1})_{km} (\bm{x})\sigma_{mj} (\bm{x}).
\end{equation}
For example, in one dimension, equation~(\ref{eq:StratSKlimit}) is 
\begin{equation}
	\label{eq:StratSKlimit1D}
	dx_t = \left (\frac{F(x_t)}{\gamma(x_t)}  -\frac{1}{2}\frac{\sigma(x_t)\sigma'(x_t)}{\gamma^2(x_t) }\right )\;dt + \frac{\sigma(x_t)}{\gamma(x_t)}\;\circ dW_t. 
\end{equation}
It follows that $\bar{\bm{S}}=0$ if the noise matrix $\bm{\sigma}$ is independent of $\bm{x}$. Note that when $\bm{\gamma}(\bm{x}) = \bm{\gamma}$ is independent of $\bm{x}$, the noise-induced drift in the It\^o SDE~(\ref{eq:SKlimit}) is zero.
 
\section{Conclusion}
\label{sec:conclusion}
We have proven convergence of solutions of a class of SDE systems in the small-mass limit.  Generalizing earlier work by several authors, the results apply in arbitrary dimension and allow to include position-dependent friction and noise coefficients, as well as colored noises with suitably scaled correlation times. Our main result (Theorem~\ref{theorem}) provides an alternative to homogenization of SDE obtained by multiscale expansion; while the latter prove convergence in distribution, our method yields stronger $L^2$-convergence. It has a wide range of physically relevant applications, including explanation of actual experiments and prediction of new effects. We have, in particular, discussed applications to Brownian motion in a diffusion gradient, thermophoresis of small particles, and Brownian motion in the presence of non-conservative forces.

\bibliographystyle{plain}
\bibliography{sde_bib2}

\end{document}